\documentclass[11pt]{amsart}
\usepackage[margin=1in]{geometry}                
\usepackage{amssymb,amsmath,amsfonts,xcolor}
\usepackage[colorlinks,
             linkcolor=blue,
             citecolor=green!70!black,
             pdfproducer={pdfLaTeX},
             pdfpagemode=None,
             bookmarksopen=true
             bookmarksnumbered=true]{hyperref}

\usepackage{tikz}
\usetikzlibrary{decorations.markings,arrows,decorations.pathreplacing}

\newcommand\arXiv[1]{\href{http:/arxiv.org/abs/#1}{\nolinkurl{arXiv:#1}}}
\newcommand\MRnumber[1]{\href{http:/www.ams.org/mathscinet-getitem?mr=#1}{\nolinkurl{MR#1}}}
\newcommand\DOI[1]{\href{http:/dx.doi.org/#1}{\nolinkurl{DOI:#1}}}
\newcommand\MAILTO[1]{\href{mailto:#1}{\nolinkurl{#1}}}

\usepackage{amsthm,cleveref}

\newtheorem{theorem}[subsection]{Theorem}

\newtheorem{proposition}[subsection]{Proposition}
\newtheorem{definition}[subsection]{Definition}

\newtheorem{lemma}[subsection]{Lemma}
\newtheorem{corollary}[subsection]{Corollary}

\newtheorem{conjecture}[subsection]{Conjecture}

\renewcommand\qedhere{\hfill\ensuremath\Box}
\newcommand\diamondhere{\hfill\ensuremath\Diamond}

\theoremstyle{definition}
\newtheorem{remarknodiamond}[subsection]{Remark}
\newenvironment{remark}{\begin{remarknodiamond}}{\diamondhere\end{remarknodiamond}}
\newtheorem{examplenodiamond}[subsection]{Example}
\newenvironment{example}{\begin{examplenodiamond}}{\diamondhere\end{examplenodiamond}}

\crefname{section}{Section}{Sections}
\crefformat{section}{#2Section~#1#3} 
\Crefformat{section}{#2Section~#1#3} 

\crefname{appendix}{Appendix}{Appendices}
\crefformat{appendix}{#2Appendix~#1#3} 
\Crefformat{appendix}{#2Appendix~#1#3}

\crefname{definition}{Definition}{Definitions}
\crefformat{definition}{#2Definition~#1#3} 
\Crefformat{definition}{#2Definition~#1#3} 

\crefname{assumption}{Assumption}{Assumptions}
\crefformat{assumption}{#2Assumption~#1#3} 
\Crefformat{assumption}{#2Assumption~#1#3} 

\crefname{lemma}{Lemma}{Lemmas}
\crefformat{lemma}{#2Lemma~#1#3} 
\Crefformat{lemma}{#2Lemma~#1#3} 

\crefname{proposition}{Proposition}{Propositions}
\crefformat{proposition}{#2Proposition~#1#3} 
\Crefformat{proposition}{#2Proposition~#1#3} 

\crefname{corollary}{Corollary}{Corollaries}
\crefformat{corollary}{#2Corollary~#1#3} 
\Crefformat{corollary}{#2Corollary~#1#3} 

\crefname{theorem}{Theorem}{Theorems}
\crefformat{theorem}{#2Theorem~#1#3} 
\Crefformat{theorem}{#2Theorem~#1#3} 

\crefname{expectation}{Expectation}{Expectations}
\crefformat{expectation}{#2Expectation~#1#3} 
\Crefformat{expectation}{#2Expectation~#1#3} 

\crefname{conjecture}{Conjecture}{Conjectures}
\crefformat{conjecture}{#2Conjecture~#1#3} 
\Crefformat{conjecture}{#2Conjecture~#1#3}

\crefname{remark}{Remark}{Remarks}
\crefformat{remark}{#2Remark~#1#3} 
\Crefformat{remark}{#2Remark~#1#3} 

\crefname{remarknodiamond}{Remark}{Remarks}
\crefformat{remarknodiamond}{#2Remark~#1#3} 
\Crefformat{remarknodiamond}{#2Remark~#1#3}

\crefname{example}{Example}{Examples}
\crefformat{example}{#2Example~#1#3} 
\Crefformat{example}{#2Example~#1#3} 

\crefname{examplenodiamond}{Example}{Examples}
\crefformat{examplenodiamond}{#2Example~#1#3} 
\Crefformat{examplenodiamond}{#2Example~#1#3}

\usepackage{dsfont}
\renewcommand\mathbb\mathds

\newcommand\bC{\mathbb C}

\newcommand\bF{\mathbb F}

\newcommand\bH{\mathbb H}

\newcommand\bJ{\mathbb J}
\newcommand\bK{\mathbb K}
\newcommand\bL{\mathbb L}

\newcommand\bN{\mathbb N}

\newcommand\bR{\mathbb R}
\newcommand\bS{\mathbb S}

\newcommand\bZ{\mathbb Z}

\newcommand\cA{\mathcal A}
\newcommand\cB{\mathcal B}
\newcommand\cC{\mathcal C}
\newcommand\cD{\mathcal D}

\newcommand\cF{\mathcal F}
\newcommand\cG{\mathcal G}

\newcommand\cO{\mathcal O}

\newcommand\cT{\mathcal T}

\newcommand\cV{\mathcal V}

\newcommand\rB{\mathrm B}

\newcommand\rH{\mathrm H}

\newcommand\rO{\mathrm O}

\newcommand\rT{\mathrm T}

\usepackage[mathscr]{euscript}

\newcommand\sV{\mathscr V}

\newcommand\sX{\mathscr X}
\newcommand\sY{\mathscr Y}

\renewcommand\H{\homology}

\usepackage[utf8]{inputenc}
\DeclareFontFamily{U}{min}{}
\DeclareFontShape{U}{min}{m}{n}{<-> udmj30}{}

\newcommand\unit{\mathbb 1}

\newcommand\mono\hookrightarrow
\newcommand\epi\twoheadrightarrow
\newcommand\isom{\overset\sim\to}
\newcommand\<\langle
\renewcommand\>\rangle
\newcommand\sminus\smallsetminus

\newcommand\id{\mathrm{id}}
\DeclareMathOperator\maps{maps}
\newcommand\pt{\{\mathrm{pt}\}}
\newcommand\op{{\mathrm{op}}}

\DeclareMathOperator{\Aut}{Aut}
\DeclareMathOperator\End{End}
\DeclareMathOperator\tr{tr}
\DeclareMathOperator\Forget{Forget}

\DeclareMathOperator\Or{Or}
\DeclareMathOperator\Her{Her}
\DeclareMathOperator\Spins{Spins}
\DeclareMathOperator\SpinStats{HerSpinStats}

\newcommand\Bord{\cat{Bord}}
\newcommand\Spans{\cat{Spans}}
\newcommand\Mod{\cat{Mod}}
\newcommand\Spaces{\cat{Spaces}}
\newcommand\Vect{\cat{Vect}}
\newcommand\SuperVect{\cat{SuperVect}}
\newcommand\MOD{\mathcal M\textsc{od}}
\newcommand\Man{\cat{Man}}
\newcommand\QCoh{\cat{Qcoh}} \newcommand\Qcoh\QCoh

\newcommand\Pres{\cat{Pres}}
\newcommand\Alg{\cat{Alg}}
\newcommand\Sets{\cat{Sets}}
\newcommand\Sch{\cat{Sch}}
\newcommand\CatAffSch{\cat{CatAffSch}}

\DeclareMathOperator\homology{H}

\DeclareMathOperator\Sym{Sym}

\DeclareMathOperator\Spec{Spec}

\newcommand\SO{\mathrm{SO}}
\newcommand\Spin{\mathrm{Spin}}
\newcommand\Pin{\mathrm{Pin}}

\newcommand\Cliff{\bC\mathrm{liff}}

\newcommand\pairing{\tikz[baseline=(b)] \draw[thick]  (0,1pt) coordinate (b) (0,0) arc (-90:90:4pt);}
\newcommand\copairing{\tikz[baseline=(b)] \draw[thick]  (0,1pt) coordinate (b) (0,0) arc (270:90:4pt);}
\newcommand\twist{\tikz[baseline=(b)] \draw[thick] (0,-1pt) coordinate (b) (-2pt,0) -- (0pt,0) .. controls +(4pt,0) and +(3pt,0) .. (4pt,5pt) .. controls +(-3pt,0) and +(-4pt,0) .. (8pt,0) -- (10pt,0);}

\newcommand\define[1]{\emph{#1}}
\newcommand\cat[1]{\textsc{#1}}

\title{Spin, statistics, orientations, unitarity}
\author{Theo Johnson-Freyd}
\address{Mathematics Department, Northwestern University \\ 2033 Sheridan Road, Evanston, IL, 60208}
\date{\today}
\email{\MAILTO{theojf@math.northwestern.edu}}
\keywords{TQFT, spin, super, categorification, torsors, Galois theory}
\subjclass[2010]{57R56 (Manifolds and cell complexes: Topological quantum field theories); 81T50 (Quantum theory: Anomalies); 14A22 (Algebraic geometry: Generalizations)}

\begin{document}
\begin{abstract}
A topological quantum field theory is \emph{Hermitian} if it is both oriented and complex-valued, and orientation-reversal agrees with complex-conjugation.  A field theory \emph{satisfies spin-statistics} if it is both spin and super, and $360^\circ$-rotation of the spin structure agrees with the operation of flipping the signs of all fermions.  We set up a framework in which these two notions are precisely analogous.  
In this framework, field theories are defined over $\textsc{Vect}_{\mathbb R}$, but rather than being defined in terms of a single tangential structure, they are defined in terms of a bundle of tangential structures over $\mathrm{Spec}(\mathbb R)$.  Bundles of tangential structures may be \'etale-locally equivalent without being equivalent, and Hermitian field theories are nothing but the field theories controlled by the unique nontrivial bundle of tangential structures that is \'etale-locally equivalent to Orientations.  This bundle owes its existence to the fact that $\pi_1^{\text{\'et}}(\mathrm{Spec}(\mathbb R)) = \pi_1\mathrm{BO}(\infty)$.  We interpret Deligne's ``existence of super fiber functors'' theorem as implying that in a categorification of algebraic geometry in which symmetric monoidal categories replace commutative rings,  $\pi_2^{\text{\'et}}(\mathrm{Spec}(\mathbb R)) = \pi_2\mathrm{BO}(\infty)$.  One finds that there are eight bundles of tangential structures \'etale-locally equivalent to Spins, one of which is distinguished; upon unpacking the meaning of a field theory with that distinguished tangential structure, one arrives at a field theory that is both Hermitian and satisfies spin-statistics.  Finally, we formulate in our framework a notion of ``reflection-positivity'' and prove that if an ``\'etale-locally-oriented'' field theory is reflection-positive then it is necessarily Hermitian, and if an ``\'etale-locally-spin'' field theory is reflection-positive then it necessarily both  satisfies spin-statistics and is Hermitian.  The latter result is a topological version of the famous Spin-Statistics Theorem.
\end{abstract}

\maketitle

\tableofcontents

\setcounter{section}{-1}

\section{Introduction}

The main result of this article is a topological version of the Spin-Statistics Theorem.  The usual Spin-Statistics Theorem (c.f.\ \cite{MR0161603}) asserts that in a unitary quantum field theory on Minkowskian spacetime, the fields of the theory live in a supervector space, the even (or bosonic) fields are integer spin representations of the Lorentz group, and the odd (or fermionic) fields are half-integer spin representations.  In other words, the spin of a ``particle'' agrees with its
parity.  
  Here ``unitarity'' is actually two conditions: a ``Hermiticity'' condition (asserting that the determinant-$(-1)$ component of the Lorentz group acts complex-antilinearly) and a ``reflection-positivity'' condition related to the requirement that the Hamiltonian of the quantum field theory have positive spectrum.

To formulate a version in the functorial setting of topological quantum field theory, we need:
\begin{itemize}
    \item to have orientations and spin structures on our ``source'' bordism category,
    \item to have complex supervector spaces in our ``target'' category, but to be able to talk about complex-antilinear maps as well as ``anti-super'' maps (i.e.\ maps that treat even and odd parts differently),
    \item to be able to link these structures on source and target categories.
\end{itemize}
We will solve all three problems by introducing generalizations of oriented and spin $\bR$-linear field theories (we generally drop the words ``topological'' and ``quantum'') that we call ``\'etale-locally-oriented'' and ``\'etale-local-spin.'' 
 \'Etale-locally-oriented and \'etale-locally-spin field theories admit a natural notion of ``reflection-positivity'' (defined in terms of a certain ``integration'' map taking in an \'etale-locally-oriented or -spin field theory and producing an unoriented $\bR$-linear field theory).  With this technology in place, our main result is the following version of the Spin-Statistics Theorem:

\begin{theorem} \label{mainthm}
  Every once-extended \'etale-locally-spin reflection-positive topological quantum field theory is Hermitian (hence unitary) and satisfies spin-statistics.
\end{theorem}

By definition, a field theory is \define{unextended} if it is defined in codimensions $0$ and $1$, and \define{once-extended} if it is defined in codimensions $0$, $1$, and $2$.
\Cref{main thm extended}, which we prove only in outline, extends \cref{mainthm} to more-than-once-extended field theories.
A similar ``Spin-Statistics Theorem'' appears in \cite[Theorem 11.3]{FreedHopkins}, but there are notable differences between the approach used there and the one used in this paper.

As a warm-up to \cref{mainthm}, in \cref{section.hermitian} we develop in detail the notions of ``\'etale-local orientation'' and ``reflection-positivity'' in the context of unextended field theories.  The following analog of \cref{mainthm} follows almost immediately from the definitions:
\begin{theorem}\label{reflection positive oriented}
  Every unextended \'etale-locally-oriented reflection-positive field theory is Hermitian.
\end{theorem}
 The parallel between \cref{mainthm,reflection positive oriented} is an indication of the second main theme of this paper, which is to argue that Hermiticity and spin-statistics phenomena arise from the same source.  Note also that we reverse part of the logic  from the standard spin-statistics theorem: as usually presented, Hermiticity is a required assumption in order to imply spin-statistics; in our version, Hermiticity and spin-statistics are both forced by reflection-positivity.

In order to define ``\'etale-locally-oriented'' manifolds, we consider local structures on manifolds that range over not (as in the case of orientations) sets, but schemes over~$\bR$.  
There are precisely two ``local structures'' that are \'etale-locally-over-$\Spec(\bR)$ isomorphic to ``orientations.'' 
The two versions of ``\'etale-local-orientations'' are usual-orientations and ``Hermitian structures''; the latter are characterized by the property that the scheme of Hermitian structures on a point is $\Spec(\bC)$ and that the restriction map $\{$Hermitian structures on $[0,1]\} \to \{$Hermitian structures on $\{0,1\}\}$ is the ``antidiagonal'' map $\Spec(\bC) \to \Spec(\bC) \times_{\Spec(\bR)}\Spec(\bC)$ sending $\lambda \mapsto (\lambda,\bar\lambda)$.  Hermitian structures owe their existence to the fact that the absolute Galois group of $\bR$ happens to be the same as the group $\pi_0\rO(\infty)$ of connected components of the orthogonal group.

Each \'etale-local-orientation leads to a version of ``\'etale-locally-oriented field theory'': in addition to the usual (unextended) oriented bordism category $\Bord_{d-1,d}^{\Or}$, there is a ``Hermitian bordism category'' $\Bord_{d-1,d}^{\Her}$ which is not a category but rather a stack of categories over $\Spec(\bR)$; the two types of field theories are symmetric monoidal functors of stacks of categories $\Bord_{d-1,d}^{\Or} \to \Vect_\bR$ and $\Bord_{d-1,d}^{\Her} \to \Vect_\bR$, where $\Vect_\bR$ is enhanced to the stack of categories $\Qcoh$.  As such, our notion of ``\'etale-locally-oriented field theory'' involves infusing both the source and target categories with $\bR$-algebraic geometry.
 The two versions unpack  to $\bR$-linear oriented field theories and to Hermitian field theories in the usual sense. 
 
Our definition of ``\'etale-locally-spin'' structures requires a categorification of (some basic notions from) real algebraic geometry.  We begin this program in \cref{section.algebraic closure}.  Our main contribution here is to categorify the notion of ``field'' and to interpret the main theorem of \cite{MR1944506} as asserting that the ``categorified algebraic closure'' of $\bR$ is not $\bC$ but rather the category $\SuperVect_\bC$ of complex supervector spaces.  (As we will use a slight modification of the main result of \cite{MR1944506}, we include a complete proof.)

\begin{remark}
  As is already apparent, we will be working both with ``fields'' in the sense of commutative algebra and ``field theories'' in the sense of physics, and English includes an unfortunate terminological conflict.  
    We don't have a good solution to this problem, but will stick to the following convention: ``field'' used as a noun means ``field in the sense of algebra''; ``field theory'' means ``(classical or quantum) functorial topological field theory in the sense of physics.''
\end{remark}

 We also prove that the extension $\Vect_\bR \mono \SuperVect_\bC$ is Galois, and use this fact to categorify the notion of ``\'etale-local.''
   There are precisely eight types of ``\'etale-locally-spin'' structures, of which one is distinguished by the following coincidence: the ``categorified absolute Galois group of $\bR$'' is canonically equivalent to the Picard groupoid $\pi_{\leq 1}\rO(\infty)$.  This distinguished version incorporates both Hermiticity and spin-statistics phenomena.  In summary, we find that the second row of the following table is a categorification of the first:\vspace{6pt}

\begin{center}\mbox{}\hspace{-2in}
 \begin{tabular}{c|c|c|c}
 Algebraic closure & Tangential structure & Galois group & Physical phenomenon
 \\ \hline \hline
    $\bR \hookrightarrow \bC$ & $\SO(d) \hookrightarrow \rO(d)$ & $\mathrm{Gal}(\bC/\bR) = \pi_0\rO(\infty)$& Hermiticity \\ \hline
    $\Vect_\bR \hookrightarrow \SuperVect_\bC$ & $\Spin(d) \to \rO(d) $ &$\mathrm{Gal}(\SuperVect_\bC/\bR) = \pi_{\leq 1}\rO(\infty)$& spin-statistics \\
  \end{tabular}\hspace{-2in}\vspace{6pt}\mbox{}
\end{center}

 Our categorification result suggests to the following conjecture:
\begin{conjecture}
  There is an infinitely-categorified version of commutative algebra, and in it the ``infinitely categorified absolute Galois group'' of $\bR$ is $\rO(\infty)$.
\end{conjecture}

\begin{remark}
The papers \cite{MR3177367,Kapranov2015} suggest that rather than $\rO(\infty)$, it is the sphere spectrum that controls supermathematics.
Very low homotopy groups cannot distinguish between various important spectra.
 The connection with topological quantum field theory focused on in this paper provides a reason to prefer $\rO(\infty)$.
\end{remark}

We prove \cref{mainthm} in \cref{section.spin statistics}, which also contains examples of various types of \'etale-locally-spin field theories.  We end the paper in \cref{section.general nonsense} by outlining how to extend our ``\'etale-locally-structured'' cobordism categories to the ``fully-extended'' $\infty$-categorical world of \cite{Lur09}.

\section{Oriented, Hermitian, and unitary field theories}  \label{section.hermitian}

This section serves as an extended warm-up to the remainder of the paper.  We will develop in a 1-categorical setting the notions of ``\'etale-locally-oriented'' and ``reflection-positive'' and  prove \cref{reflection positive oriented}, which asserts that \'etale-locally-oriented reflection-positive topological quantum field theories are necessarily Hermitian.

The functorial framework for quantum field theory, as formulated in \cite{MR1001453,MR2079383}, is well-known.  Fix a dimension $d$ and construct a symmetric monoidal category $\Bord_{d-1,d}$ whose objects are $(d-1)$-dimensional closed smooth manifolds, morphisms are $d$-dimensional smooth cobordisms up to isomorphism, and the symmetric monoidal structure is disjoint union.  An (unextended) \define{unoriented} or \define{unstructured $\bR$-linear $d$-dimensional functorial topological quantum field theory} is a symmetric monoidal functor $\Bord_{d-1,d} \to \Vect_\bR$.  We will henceforth drop the words ``functorial topological quantum.''

In general, one does not care simply about unstructured field theories.  Let $\Man_d$ denote the site of $d$-dimensional (possibly open) manifolds and local diffeomorphisms, with covers  the surjections. 
If $\sX$ is a category with limits, an \define{$\sX$-valued local structure} is a sheaf $\cG : \Man_d \to \sX$.
A local structure is \define{topological} if it takes isotopic (among local diffeomorphisms) maps of manifolds to equal morphisms in $\sX$. 

The reason for considering local structures valued in general categories is because, in examples, the collection of $\cG$-structures on a manifold $M$ is not just a set but carries more algebraic or analytic structure.  For example, the paper \cite{MR2742432} requires local structures valued in supermanifolds.  We will focus on the case when $\cG$ is valued in the category $\Sch_\bR$ of schemes over $\bR$.  (In fact, all of our examples will take values in the subcategory $\cat{AfSch}_\bR$ of affine schemes.)

The following is an easy exercise:
\begin{lemma}\label{set theory cob hyp}
  Suppose $d\geq 1$.  There are precisely two isotopy classes of local diffeomorphisms $\bR^d \to \bR^d$ (the identity and orientation reversal), and so if $\cG$ is an $\sX$-valued topological local structure, then $\cG(\bR^d)$ has an action by $\bZ/2$.  The assignment $\cG \mapsto \cG(\bR^d)$ gives an equivalence of categories between the category of $\sX$-valued topological local structures and the category $\sX^{\bZ/2}$ of objects in $\sX$ equipped with a $\bZ/2$-action. \qedhere
\end{lemma}
\begin{example}\label{orientations}
  The topological local structure $\cG_X$ corresponding to a $\bZ/2$-set $X \in \Sets^{\bZ/2}$ can be constructed as follows.  For any manifold $M$, let $\Or_M \to M$ denote the orientation double cover; then $\cG_X(M) = \maps_{\bZ/2}(\Or_M,X)$, where $\maps_{\bZ/2}$ denotes continuous $\bZ/2$-equivariant functions.  If $\sX$ has limits, then for $X \in \sX^{\bZ/2}$ the formula ``$\maps_{\bZ/2}(\Or_M,X)$'' continues to make sense, and again defines the topological local structure corresponding to $X$.
  
  The most important example is when $X = \bZ/2$ is the \define{trivial $\bZ/2$-torsor} given by the translation action of $\bZ/2$ on itself.  Then $\cG_{\bZ/2} = \Or$ is the sheaf $\Or(M) = \{\text{orientations of $M$}\}$.
\end{example}

 Given a $\Sets$-valued topological local structure $\cG$, there is a \define{$\cG$-structured bordism category} $\Bord^\cG_{d-1,d}$, an object of which consists of a closed $(d-1)$-manifold $N$ together with an element of $\cG(N\times \bR)$, and whose morphisms are $d$-dimensional cobordisms similarly equipped with $\cG$-structure.  
  If $\cG$ is a $\Sets$-valued topological local structure, a \define{$\cG$-structured $\bR$-linear $d$-dimensional field theory} is a symmetric monoidal functor $\Bord^\cG_{d-1,d} \to \Vect_\bR$.
It will be useful to unpack the construction of $\Bord_{d-1,d}^\cG$ in order to have a more explicit description of $\cG$-structured field theories.  The following logic is used in \cite[Section 3.2]{Lur09} to reduce the ``$\cG$-structured Cobordism Hypothesis'' to the unstructured case; see also \cite[Section 3.5]{Schommer-Pries:thesis}.

Let $\Spans(\Sets)$ denote the symmetric monoidal category whose objects are sets and morphisms are isomorphism classes of \define{correspondences}, i.e.\ diagrams of shape $X \leftarrow A \to Y$; composition is by fibered product and the symmetric monoidal structure is by cartesian product.  A \define{$\cG$-structured classical field theory} is a symmetric monoidal functor $\Bord_{d-1,d}^\cG \to \Spans(\Sets)$.  Every $\Sets$-valued topological local structure $\cG$ defines an unstructured classical field theory $\widetilde\cG : \Bord_{d-1,d} \to \Spans(\Sets)$:
$$\begin{tikzpicture}[baseline=(middle)]
  \coordinate (middle) at (0,0);
  \draw[very thick,dotted] (-.5,.7) arc (90:270:.35 and .7);
  \draw[very thick,dotted] (-.5,.7) arc (90:-90:.35 and .7);
  \draw[thin] (-.25,.6) arc (90:270:.3 and .6);
  \draw[thin] (-.25,.6) arc (90:-90:.3 and .6);
  \draw[thick,dotted] (0,.5) arc (90:270:.25 and .5);
  \draw[very thick,dotted] (0,.5) arc (90:-90:.25 and .5);
  \draw[very thick,dotted] (4,.5) arc (90:270:.25 and .5);
  \draw[thick,dotted] (4,.5) arc (90:-90:.25 and .5);
  \draw[thin] (4.25,-.525) arc (-90:270:.275 and .525);
  \draw[very thick,dotted] (4.5,-.6) arc (-90:270:.3 and .6);
  \draw[thick] (-.5,.7) -- (0,.5) .. controls +(1,-.4) and +(-1,0) .. (2,1) .. controls +(1,0) and +(-1,0) .. (4,.5) .. controls +(.1,0) and +(-.2,-.1) .. (4.5,.6);
  \draw[thick] (-.5,-.7) -- (0,-.5) .. controls +(1,.4) and +(-1,0) .. (2,-1) .. controls +(1,0) and +(-1,0) .. (4,-.5) .. controls +(.1,0) and +(-.2,.1) .. (4.5,-.6);
  \draw[thick] (1.5,0) .. controls +(.3,.2) and +(-.3,.2) .. (2.5,0);
  \draw[thick] (1.35,.1) -- (1.5,0) .. controls +(.3,-.2) and +(-.3,-.2) .. (2.5,0) -- (2.65,.1);
  \draw[thick,decoration={brace,amplitude=3},decorate] (.1,-.9) -- coordinate (N1) (-.6,-.9) (N1) +(0,-.1) node[anchor=north]  {$\scriptstyle N_1$};
  \draw[thick,decoration={brace,amplitude=3},decorate] (4.6,-.9) -- coordinate (N2) (3.9,-.9) (N2) +(0,-.1) node[anchor=north]  {$\scriptstyle N_2$};
  \node at (2,-.5) {$\scriptstyle M$};
\end{tikzpicture}
\quad\overset{\textstyle \widetilde\cG}\longmapsto
\begin{tikzpicture}[baseline=(middle),scale=2]
  \path (0,0) node (M) {$ \{\cG\text{-structures on }M\}$}
  (-1.25,-1) node (N1) {$ \{\cG\text{-structures on }N_1\}$}
  (1.25,-1) node (N2) {$ \{\cG\text{-structures on }N_2\}$}
  (0,-.5) coordinate (middle);
  \draw[->] (M) -- node[anchor=south east] {restrict} (N1);
  \draw[->] (M) -- node[anchor=south west] {restrict} (N2);
\end{tikzpicture}
$$
Functoriality for $\widetilde\cG : \Bord_{d-1,d} \to \Spans(\Sets)$ follows from the sheaf axiom for~$\cG$.

Definition-unpacking implies:
\begin{lemma}\label{description of G-structures}
  Let $\Spans(\Sets;\Vect_\bR)$ denote the symmetric monoidal category whose objects are pairs $(X, V)$ where $X \in \Sets$and $V$ is a vector bundle over $X$, and for which a morphism from $(X,V)$ to $(Y, W)$ is an isomorphism class of diagrams $X \overset f \leftarrow A \overset g \to Y$ together with a vector bundle map $f^*V \to g^*W$.  Then a $\cG$-structured field theory is the same data as a choice of lift:
  
  \mbox{}\hfill $ \begin{tikzpicture}[baseline=(LR.base)]
   \path coordinate (UL) +(3,0) node (UR) {$\Spans(\Sets;\Vect_\bR)$} +(0,-1.5) node (LL) {$\Bord_{d-1,d}$} +(3,-1.5) node (LR) {$\Spans(\Sets)$}
    ;
   \draw[->] (LL) -- node[auto] {$\scriptstyle \widetilde\cG$} (LR); \draw[->] (UR) -- node[auto] {\scriptsize {Forget} the $\Vect_\bR$-data} (LR);
   \draw[->,dashed] (LL) -- (UR);
  \end{tikzpicture} $ \qedhere
\end{lemma}

Suppose that $\cG$ is a topological local structure valued not in $\Sets$ but in $\Sch_\bR$.  Our strategy will be to take \cref{description of G-structures} as the model for the definition of ``$\cG$-structured field theory.''  To do this, note that $\Vect_\bR$ is naturally an object of $\bR$-algebraic geometry.  Indeed, there is a stack of categories on $\Sch_\bR$, namely $\Qcoh: \Spec(A) \mapsto \cat{Mod}_A$, whose category of global sections  is nothing but $\Qcoh(\Spec(\bR)) = \Vect_\bR$.  We can therefore define:

\begin{definition} \label{schemes structured field theory}
  Let $\cG$ be a topological local structure valued in schemes over $\bR$, thought of as a ``classical field theory'' $\widetilde\cG: \Bord_{d-1,d} \to \Spans(\Sch_\bR)$.
  Let $\Spans(\Sch_\bR;\Qcoh)$ denote the symmetric monoidal category whose objects are pairs $(X,V)$ where $X$ is a scheme over $\bR$ and $V \in \Qcoh(X)$, in which a morphism from $(X,V)$ to $(Y,W)$ is (an isomorphism class of) a correspondence of schemes $X \overset f \leftarrow A \overset g \to Y$ together with a map of quasicoherent sheaves $f^*V \to g^*W$, in which composition is by fibered product, and in which the symmetric monoidal structure is $\times_{\Spec(\bR)}$.
    A \define{$\cG$-structured field theory} is a choice of lift:
       $$ \begin{tikzpicture}[baseline=(LR.base)]
   \path coordinate (UL) +(3,0) node (UR) {$\Spans(\Sch_\bR;\Qcoh)$} +(-1,-1.5) node (LL) {$\Bord_{d-1,d}$} +(3,-1.5) node (LR) {$\Spans(\Sch_\bR)$}
    ;
   \draw[->] (LL) -- node[auto] {$\scriptstyle \widetilde\cG$} (LR); \draw[->] (UR) -- node[auto] {\scriptsize {Forget} the $\Qcoh$-data} (LR);
   \draw[->,dashed] (LL) -- (UR);
  \end{tikzpicture}$$
\end{definition}

  Any topological local structure $\cG$ valued in $\Sets$ defines a topological local structure, which we will also call $\cG$, valued in $\Sch_\bR$, via the symmetric monoidal inclusion $\Sets \hookrightarrow \Sch_\bR, S \mapsto S \times \Spec(\bR)$.  In this case, the notion of ``$\cG$-structured field theory'' from \cref{schemes structured field theory} agrees with the usual notion in terms of symmetric monoidal functors $\Bord_{d-1,d}^\cG \to \Vect_\bR$, since $\Qcoh(S\times\Spec(\bR)) = \{\text{real vector bundles on $S$}\}$.
  
 We will focus on four examples of topological local structures $\cG$ valued in $\Sch_\bR$, two of which come from topological local structures valued in $\Sets$.  We will unpack a bit about the values of $\cG$-structured field theories in all four cases to make everything explicit.
  
\begin{example} \label{eg.unstructured field theory}
  An \define{unstructured} or \define{unoriented} field theory is a ``${\Spec(\bR)}$-structured'' one, where ${\Spec(\bR)}(M) = \Spec(\bR)$ for all manifolds $M$.  Let $Z$ be an unstructured field theory.   If $M$ a closed $d$-dimensional manifold, then $Z(M) \in \cO(\Spec(\bR)) = \bR$.  If $N$ is a closed $(d-1)$-dimensional manifold, then $Z(N) \in \Qcoh(\Spec(\bR)) = \Vect_\bR$.  Consider the \define{macaroni} cobordisms $N \times \pairing : N\sqcup N \to \emptyset$ and $N \times \copairing : \emptyset \to N\sqcup N$.  The first defines a symmetric pairing $Z(N \times \pairing) : Z(N) \otimes Z(N) \to \bR$ and the second a symmetric copairing $\bR \to Z(N) \otimes Z(N)$.  The \define{zig-zag equations} $\tikz[baseline=(b)] \draw[thick] (0,.5pt) coordinate (b) (0,0) -- (4pt,0) arc (-90:90:2pt) arc (270:90:2pt) -- +(4pt,0); = \tikz[baseline=(b)] \draw[thick] (0,.5pt) coordinate (b) (0,0) .. controls +(4pt,0) and +(-4pt,0) .. ++(8pt,8pt);$ and $\tikz[baseline=(b),xscale=-1] \draw[thick] (0,.5pt) coordinate (b) (0,0) -- (4pt,0) arc (-90:90:2pt) arc (270:90:2pt) -- +(4pt,0); = \tikz[baseline=(b),xscale=-1] \draw[thick] (0,.5pt) coordinate (b) (0,0) .. controls +(4pt,0) and +(-4pt,0) .. ++(8pt,8pt);$ require this pairing and copairing to be inverse to each other, and are equivalent to making $V = Z(N)$ into a symmetrically-self-dual vector space over $\bR$, i.e.\ we have $\varphi : V \isom V^*$ with $\varphi^* \circ \varphi = \id_V$.
\end{example}

\begin{example} \label{eg.oriented field theory}
  An \define{oriented} field theory is one with topological local structure $\Or = \cG_{\bZ/2}$ from \cref{orientations}, thought of as being valued in $\Sch_\bR$ via $S \mapsto S \times \Spec(\bR)$.  Orientations are distinguished among all topological local structures by \cref{set theory cob hyp}: they correspond to the trivial $\bZ/2$-torsor.  We will review the basic structure enjoyed by an oriented field theory $Z$.
  
  Let $M$ be a connected closed $d$-dimensional manifold.  Then $Z(M)$ is a function on  $\Or(M)\times \Spec(\bR)$.  If $M$ is unorientable, then $\Or(M) = \emptyset$ and $Z(M)$ is no data.  If $M$ is orientable, then $\Or(M) \times \Spec(\bR) \cong \Spec(\bR) \sqcup \Spec(\bR)$, the two points corresponding to the two orientations of $M$, and $Z(M)$ is an element of $\cO(\Spec(\bR) \sqcup \Spec(\bR)) = \bR \times \bR$, i.e.\ a pair of numbers (indexed by the two orientations of $M$).
  
  Suppose now that $N$ is a closed connected $(d-1)$-dimensional manifold.  Again if $N$ is unorientable, $\Or(N)$ is empty and $Z$ assigns no data.  If $N$ is orientable,  $Z(N)$ is a sheaf on $\Or(N) \times \Spec(\bR) \cong \Spec(\bR) \sqcup \Spec(\bR)$, i.e.\ a pair $(V,V')$ of real vector spaces, one for each orientation of $N$.  These vector spaces are not independent.  Rather, the macaroni cobordisms $N \times \pairing : N\sqcup N \to \emptyset$ and $N \times \copairing : \emptyset \to N\sqcup N$ each admit two orientations, which induce orientations of their boundaries such that the two copies of $N$ have opposite orientations.  Some definition-unpacking shows that the data of $Z(N \times \pairing)$ is nothing but a linear map $V \otimes_\bR V' \to \bR$, and the data of $Z(N \times \copairing)$ is a linear map $\bR \to V \otimes_\bR V'$.  The zig-zag equations assert that $Z(N \times \pairing)$ and $Z(N \times \copairing)$ make $V$ and $V'$ into dual vector spaces.
\end{example}

\begin{example} \label{eg.hermitian field theory}
  \Cref{set theory cob hyp} distinguishes a second topological local structure valued in $\Sch_\bR$.  Specifically, there is a canonical nontrivial $\bZ/2$-torsor over $\Spec(\bR)$, namely $\Spec(\bC)$ with the complex conjugation action.  We will suggestively write $\Her : \Man_d \to \Sch_\bR$ for this topological local structure, and call $\Her(M)$ the scheme of \define{Hermitian structures} on $M$.  
  One easily sees that for any manifold $M$, $$ \Her(M) = \Or(M) \underset{\bZ/2}\times \Spec(\bC) $$
  where $\bZ/2$ acts on $\Or(M)$ by orientation reversal and on $\Spec(\bC)$ by complex conjugation, and $\times_{\bZ/2}$ denotes the coequalizer of these actions.  
  A Hermitian field theory is \define{\'etale-locally-oriented} in the sense that that $\Her$ and $\Or$ are both valued in schemes \'etale over $\Spec(\bR)$ and are \'etale-locally isomorphic as topological local structures over $\Spec(\bR)$, since they
  pull back to isomorphic topological local structures along $\Spec(\bC) \to \Spec(\bR)$.  Since there are precisely two $\bZ/2$-torsors over $\Spec(\bR)$, there are precisely two topological local structures \'etale-locally isomorphic to $\Or$, i.e.\ precisely two kinds of \'etale-locally-oriented field theory.
  
  We now justify the name ``Hermitian.'' Suppose that $Z$ is a $\Her$-structured field theory and $M$ is a closed $d$-dimensional manifold.   If $M$ is not orientable, then $\Her(M) = \emptyset$ is the empty scheme and $Z(M)$ is no data.  If $M$ is orientable and non-empty, then $\Her(M)$ is noncanonically isomorphic to the disjoint union of $2^{|\pi_0 M|-1}$  copies of $\Spec(\bC)$.  In particular, if $M$ is connected and orientable, then either orientation of $M$ determines an isomorphism $\Her(M) \cong \Spec(\bC)$.  Thus, either choice of orientation identifies $Z(M) \in \cO(\Her(M))$ with a complex number.  The two choices of orientation determine isomorphisms that differ by complex conjugation.  So one can think of $Z$ as assigning to every oriented manifold a complex number, subject to the condition that orientaiton-reversal agrees with complex conjugation.  Finally, if $M = \emptyset$, then $\Her(M) = \Spec(\bR)$ and $Z(M) = 1$.
  
    Suppose now that $N$ is a closed connected $(d-1)$-dimensional manifold.  Again if $N$ is unorientable, then $\Her(N) = \emptyset$ and $Z(N)$ is no data.  If $N$ is orientable, $Z(N)$ is a vector bundle on $\Her(N) \cong \Spec(\bC)$, i.e.\ a complex vector space.  The values of the macaroni $Z(N \times \pairing)$ and $Z(N \times \copairing)$ now are bundles of linear maps over $\Her(N \times \pairing) \cong \Her(N\times \copairing) \cong \Spec(\bC)$.  The domain and codomain of $Z(N \times \pairing)$ are given by pulling back $Z(N\sqcup N)$ and $Z(\emptyset)$ along the restrictions $\Her(N \times \pairing) \to \Her(N\sqcup N) = \Her(N) \times_{\Spec\bR}\Her(N)$ and $\Her(N\times \pairing) \to \Her(\emptyset)$, and similarly for $Z(N\times \copairing)$.  Unpacking gives:
\begin{align*}
 Z(N \times \pairing) & \in  \hom_\bR\bigl(Z(N\times \pt)\otimes_\bR Z(N\times \pt),\bR\bigr) \otimes_\bR \bC \\
 Z(N \times \copairing) & \in  \hom_\bR\bigl(\bR,Z(N \times \pt)\otimes_\bR Z(N \times \pt)\bigr) \otimes_\bR \bC 
\end{align*}
The restriction map $\Spec(\bC) = \Her(N \times \pairing) \to \Her(N) \times_{\Spec\bR}\Her(N) = \Spec(\bC) \times_{\Spec(\bR)} \Spec(\bC)$ is the ``antidiagonal'' map $\lambda \mapsto (\lambda,\bar\lambda)$, and so $Z(N \times \pairing)$ is a sesquilinear pairing on $Z(N)$.
It follows from the zig-zag equations that $Z(N \times \pairing)$ and $Z(N \times \copairing)$ identify the $\bC$-linear dual vector space $Z(N)^*$ to $Z(N) \in \Vect_\bC$ with the complex conjugate space $\overline{Z(N)}$.  Finally, the symmetry of $N \times \pairing$ translates into the requirement that the sesquilinear pairing on $Z(N)$ is symmetric, or equivalently the isomorphism $\varphi : Z(N)^* \isom \overline{Z(N)}$ satisfies $\bar\varphi^* \circ \varphi = \id$.  It is in this sense that Hermitian field theories are ``Hermitian.''
\end{example}

\begin{example} \label{eg.complex field theory}
  In addition to $\Her : \Man_d \to \Sch_\bR$, there is another topological local structure whose value on $\bR^d$ is $\Spec(\bC)$, namely the one corresponding via \cref{set theory cob hyp} to $\Spec(\bC)$ with the trivial $\bZ/2$-action.  We will simply call this topological local structure ``$\Spec(\bC)$.''  It satisfies ${\Spec(\bC)}(M) = \Spec(\bC)^{\pi_0M}$ for every manifold $M$.  When one unpacks the notion of ``${\Spec(\bC)}$-structured field theory,'' one finds that they are nothing but \define{complex-linear} unstructured field theories.  For example, the values of $\Spec(\bC)$-structured field theories on closed connected $(d-1)$- and $d$-dimensional manifolds are objects of $\QCoh(\Spec(\bC)) = \Vect_\bC$ and elements of $\cO(\Spec(\bC)) = \bC$ respectively.
\end{example}

\Cref{eg.hermitian field theory} provided one of two reasons why Hermitian field theories are distinguished: they correspond to the unique nontrivial torsor over $\Spec(\bR)$ for the group $\bZ/2 = \pi_0\hom_{\Man_d}(\bR^d,\bR^d)$.  \Cref{reflection positive oriented} provides the second reason, by asserting that of the two types of \'etale-locally-oriented field theories, only Hermiticity is compatible with reflection-positivity.  We now define reflection-positivity and prove \cref{reflection positive oriented}.

\begin{definition}\label{defn.reflection positive}
  A $d$-dimensional unstructured (i.e.\ $\Spec(\bR)$-structured) field theory $Z : \Bord_{d-1,d} \to \Vect_\bR$ is \define{reflection-positive} of the nondegenerate symmetric pairing $Z(N \times \pairing) : Z(N) \otimes Z(N) \to \bR$ is positive-definite for every closed $(d-1)$-dimensional manifold $N$.
\end{definition}

Most of the physics literature, including the original definition of functorial topological field theory from \cite{MR1001453}, includes Hermiticity directly in the definition of ``quantum field theory.''  As such, reflection-positivity is usually  posed as the requirement that the Hermitian form on the complex vector space $Z(N)$ should be positive-definite.  For non-topological quantum field theories defined on Minkowski $\bR^{d-1,d}$, reflection-positivity is a stronger condition assuring the existence of an analytic continuation to ``imaginary time'' $\bR^{d-1} \times i\bR_{\geq 0}$, and ``reflection'' refers to reflection in the  time axis.  Positive-definiteness of the Hilbert space is what remains when interpreting this stronger condition for topological field theories.

From the point of view of this paper, the non-Hermitian version of reflection-positivity in \cref{defn.reflection positive} is the most primitive.  The Hermitian version arises as follows.  Suppose first that $Z$ is not Hermitian but oriented.  One can produce an unstructured field theory $\int_{\Or}Z$ from $Z$ by integrating out the choice of orientation:
$$ \int_{\Or} Z : M \mapsto \int_{\sigma \in \Or(M)}Z(M,\sigma) $$
Here the integral is a finite sum of numbers when $M$ is $d$-dimensional and a finite direct sum when $\dim M < d$.  In particular, for $N$ a connected $(d-1)$-dimensional manifold, $(\int_{\Or}Z)(N) = Z(N) \oplus Z(N)^*$ with the obvious symmetric pairing.

Let $Z^*$ denote the orientation-reversal of the field theory $Z$.  There is a canonical equivalence $\int_{\Or}Z \cong \int_{\Or} Z^*$.  It follows that ``$\int_{\Or}$'' makes sense not just for oriented field theories but for any \'etale-locally-oriented field theory.  Indeed, suppose $Z$ is not oriented but Hermitian.  Using the isomorphism $\Her \times_{\Spec(\bR)} \Spec(\bC) \cong \Or \times \Spec(\bC)$, one sees that the base-changed field theory $Z_\bC = Z \otimes_\bR \bC$ is naturally oriented and $\bC$-linear, and so $\int_{\Or}Z_\bC$ makes sense as a $\bC$-linear unstructured field theory.  But the Hermiticity of $Z$ defines a Galois action on $\int_{\Or} Z_\bC$, describing how to descend it to an $\bR$-linear unstructured field theory $\int_{\Or} Z$.  One finds that, for $Z$ a Hermitian field theory and $N$ a connected $(d-1)$-dimensional manifold, $(\int_{\Or} Z)(N)$ is nothing but the underlying real vector space of the Hermitian vector space $Z(N)$; the symmetric pairing is twice the real part of the Hermitian pairing on $Z(N)$.

The usual notion of ``reflection-positive'' is then captured by the following:

\begin{definition}\label{defn.reflection positive for oriented}
  An \'etale-locally-oriented field theory $Z$ is \define{reflection-positive} if the unoriented field theory $\int_{\Or}Z$ is reflection-positive.
  A field theory is \define{unitary} if it is reflection-positive and Hermitian.
\end{definition}

With this notion, the proof of \cref{reflection positive oriented} is immediate:

\begin{proof}[Proof of \cref{reflection positive oriented}]
  If $V$ is a non-zero real vector space, $V\oplus V^*$ is never positive-definite.
\end{proof}

\begin{remark}\label{remark.reflection positive for complex}
  One can also integrate a $\Spec(\bC)$-structured field theory to a $\Spec(\bR)$-structured one.  One finds that for $c\in \cO(\Spec(\bC))$, $\int_{\Spec(\bC)}c = 2\operatorname{Re}(c)$, and that the integral of a complex vector space $V \in \Vect_\bC$ is the underlying real vector space of $V$.  If $Z$ is a $\Spec(\bC)$-structured field theory, then $Z(N \times \pairing)$ is a $\bC$-linear symmetric pairing on the complex vector space $Z(N)$, and $\int_{\Spec(\bC)}Z(N\times \pairing)$ is twice its real part, thought of as a symmetric pairing on the real vector space $\int_{\Spec(\bC)}Z(N)$.  The real part of a complex-linear symmetric pairing is never positive-definite.
\end{remark}

\section{A categorified Galois extension} \label{section.algebraic closure}

\Cref{section.hermitian} illustrated the important role that algebraic geometry and Galois theory play in explaining the origin of Hermitian phenomena in quantum field theory.  The goal of this section and the next is to tell a similar story concerning ``super'' phenomena of fermions and spinors.  Explicitly, $\bC$ appeared because it is the algebraic closure of $\bR$.  This section will explain that $\SuperVect_\bC$ is the ``categorified algebraic closure'' of $\Vect_\bR$.   This is essentially the ``existence of super fiber functors'' theorem from \cite{MR1944506}.  We state this result as \cref{thm.Deligne} and provide details of its proof, as our phrasing is somewhat different from that of \cite{MR1944506}.

A convenient setting for ``categorified $\bR$-linear algebra'' is provided by the bicategory $\Pres_\bR$ of $\bR$-linear locally presentable categories, $\bR$-linear cocontinuous functors, and natural transformations: direct sums play the role of addition and quotients play the role of subtraction.   
 Two of the many ways that $\Pres_\bR$ is convenient are that it admits all limits and colimits \cite{BirdThesis} and that it has a natural symmetric monoidal structure $\boxtimes = \boxtimes_\bR$ satisfying a hom-tensor adjunction \cite{MR651714}.  The unit object for $\boxtimes$ is $\Vect_\bR$.  Basic examples of $\bR$-linear locally presentable categories include the categories $\Mod_A$ of $A$-modules for any $\bR$-algebra $A$; the tensor product enjoys $\Mod_A \boxtimes \Mod_B \simeq \Mod_{A\otimes B}$.

\begin{definition}\label{defn.cat com alg}
  A \define{categorified commutative $\bR$-algebra} is a symmetric monoidal object in $\Pres_\bR$.
\end{definition}

We embed non-categorified commutative $\bR$-algebras among categorified commutative $\bR$-algebras with the following lemma, whose proof is a straightforward exercise (see \cite[Proposition 2.3.9]{MR3097055}):
\begin{lemma}\label{lemma.every affine is categorifiable}
  The assignment taking a commutative $\bR$-algebra $R$ to the categorified commutative $\bR$-algebra $(\Mod_R,\otimes_R)$ and an $\bR$-algebra homomorphism $f : R \to S$ to extension of scalars $\otimes_R S : \Mod_R \to \Mod_S$ defines a fully faithful embedding of the category of commutative $\bR$-algebras into the bicategory of categorified commutative $\bR$-algebras. \qedhere
\end{lemma}

We turn now to categorifying the notion of ``algebraic closure.''  Algebraic closures of fields are determined by a weak universal property ranging over only finite-dimensional algebras.  Summarizing the story over $\bR$, we have:

\begin{lemma} \label{lemma.definition of field}
  \begin{enumerate} \label{lemma.RversusC} \setcounter{enumi}{-1}
    \item \label{item.fd} $\bC$ is a non-zero finite-dimensional commutative $\bR$-algebra.
    \item \label{item.field} Every map $\bC \to A$ of non-zero  finite-dimensional commutative $\bR$-algebras is an injection.
    \item \label{item.closed} If $A $ is a non-zero  finite-dimensional commutative $\bR$-algebra, then there exists a map $A \to \bC$ of commutative $\bR$-algebras.
    \item \label{item.unique} Items (\ref{item.fd}--\ref{item.closed}) determine $\bC$ uniquely up to non-unique isomorphism.
    \qedhere
  \end{enumerate}
\end{lemma}
Of course, (\ref{item.fd}--\ref{item.field}) are equivalent to the statement that $\bC$ is a \define{field}, and (\ref{item.closed}) is equivalent to the statement that $\bC$ is \define{algebraically closed}.  We categorify these notions in turn.

\begin{definition}\label{defn.finite dimensional}
  A \define{strongly generating set} in an $\bR$-linear locally presentable category $\cC$ is a set of objects in $\cC$ that generate $\cC$ under  colimits.  The category $\cC$ is \define{finite-dimensional} if it admits a finite strongly generating set $\{C_1,\dots,C_n\}$ such that all hom-spaces between generators $\hom(C_i,C_j)$ are finite-dimensional and moreover every generator $C_i$ is \define{compact projective} in $\cC$, in the sense that $\hom(C_i,-) : \cC \to \Vect_\bR$ is cocontinuous.
  
  A categorified commutative $\bR$-algebra $(\cC,\otimes_\cC,\dots)$ is \define{finite-dimensional} as a categorified commutative $\bR$-algebra if the underlying $\bR$-linear category of $\cC$ is finite-dimensional and moreover every projective object $P \in \cC$ is {dualizable}.
\end{definition}

  Compact projectivity, sometimes called ``tininess,'' is a strong but reasonable finiteness condition to impose on an object.  There are many definitions of ``projective'' that agree for abelian categories but diverge for locally presentable but not necessarily abelian categories; ours is one of the stronger possible choices.  If $\cC$ is a finite-dimensional $\bR$-linear locally presentable category, then $\cC$ is automatically equivalent to the category $\Mod_A$ of modules for a finite-dimensional associative algebra $A$ (e.g.\ one can  take $A = \End(\bigoplus_i C_i)$).
  
  Finite-dimensionality as a categorified algebra is stronger than just finite-dimensionality of the underlying category.  The condition that compact projectivity implies dualizability expresses a compatibility between ``internal'' and ``external'' notions of finite-dimensionality in a symmetric monoidal category, which otherwise might badly diverge \cite{MR1711565}.  Indeed, $P \in \Mod_A$ is compact projective exactly when the functor $\otimes_\bR P : \Vect_\bR \to \Mod_A$ has a right adjoint of the form $\otimes_A P^\vee$ for some left $A$-module $P^\vee$, whereas, for $(\cC,\otimes_\cC,\dots)$ a symmetric monoidal category, $P \in \cC$ is dualizable when the functor $\otimes P : \cC \to \cC$ has a right adjoint of the form $\otimes P^*$ for some $P^* \in \cC$.
  
  To check that $(\cC,\otimes_\cC,\dots)$ is finite-dimensional as a categorified commutative $\bR$-algebra, it suffices to check that the underlying $\bR$-linear category $\cC$ is finite-dimensional and that each generator $C_i$ is dualizable.

\Cref{defn.finite dimensional} explains how to categorify item (\ref{item.fd}) from \cref{lemma.definition of field}.  With it in hand, we may   categorify the notion of ``algebraically closed field'' by following items (\ref{item.field}) and (\ref{item.closed}):

\begin{definition}\label{defn.categorified field}
 A  \define{categorified field} is a non-zero finite-dimensional categorified commutative $\bR$-algebra $(\cC,\otimes_\cC,\dots)$ such that every 1-morphism $(\cC,\otimes_\cC,\dots) \to (\cD,\otimes_\cD,\dots)$ of non-zero categorified commutative $\bR$-algebras is faithful and injective on isomorphism classes of objects.
  
   A finite-dimensional categorified field $(\cC,\otimes,\dots)$ is \define{algebraically closed} if for every non-zero finite-dimensional categorified commutative $\bR$-algebra $(\cB,\otimes,\dots)$, there exists a 1-morphism $F : (\cB,\otimes_\cB,\dots) \to (\cC,\otimes_\cC,\dots)$ of categorified commutative $\bR$-algebras.
\end{definition}

\begin{lemma} \label{fields are categorified fields}
  A finite-dimensional commutative $\bR$-algebra $R$ is a field if and only if $(\Mod_R,\otimes_R,\dots)$ is a categorified field.
\end{lemma}

\begin{proof}It is clear that if $(\Mod_R,\otimes_R,\dots)$ is a categorified field, then $R$ is a field, simply by using the faithfulness assumption and 1-morphisms to categorified algebras of the form $(\Mod_S,\otimes_S,\dots)$.  

Conversely, suppose $R$ is a field and $F : \Mod_R \to \cC$ is any $\bR$-linear functor.  Suppose $F$ is not faithful.  Then there is a non-zero morphism $f : X \to Y$ in $\Mod_R$ with $F(f) = 0$.  Using the fact that in $\Mod_R$ all exact sequences split, one can show that $F(\operatorname{im}(f)) = 0$, from which it follows that $F(R) = 0$.  If $F$ is symmetric monoidal, $F(R) \cong \unit_\cC$ is the monoidal unit in $\cC$, and so $\cC$ is the zero category.  This verifies the faithfulness condition in \cref{defn.categorified field}.

Suppose that $(\cC,\otimes_\cC,\dots)$ is a finite-dimensional categorified commutative algebra over $\bR$, and let $\unit_\cC$ denote its monoidal unit.  Any $\lambda \in \End_\cC(\unit_\cC)$ defines a natural endomorphism of the identity functor on $\cC$ via $\lambda|_X = \lambda \otimes \mathrm{id}_X : X = \unit_\cC \otimes_\cC X \to \unit_\cC \otimes_\cC X = X$, and clearly $\lambda|_\unit = \lambda$.  Since $\cC$ is finite-dimensional, it is equivalent to $\Mod_A$ for a finite-dimensional associative algebra~$A$; then the algebra of natural endomorphisms of the identity functor is nothing but the center $Z(A) \subseteq A$.  If follows that $\End_\cC(\unit_\cC) \subseteq Z(A)$ is finite-dimensional.  Suppose that $\unit_\cC \in \cC$ corresponded to an infinite-dimensional $A$-module $M_A$.  
Then $\End_A(M_A) = \End_\cC(\unit_\cC)$ would be infinite-dimensional, as it is the subalgebra of $\End_\bR(M)$ cut out by finitely many equations (imposing compatibility with multiplication by a basis in the finite-dimensional algebra $A$).
  It follows that $\unit_\cC$ corresponds to a finite-dimensional $A$-module, and so $\unit_\cC$ is a \define{compact} object in $\cC$ in the sense that $\hom_\cC(\unit_\cC,-) : \cC \to \Vect_\bR$ preserves infinite direct sums.

If $R$ is a field, every object in $\Mod_R$ is isomorphic to $R^{\oplus \alpha}$ for some cardinal $\alpha$.
Let $F : (\Mod_R,\otimes_R,\dots) \to \cC$ be a cocontinuous symmetric monoidal functor.  On objects it takes $R \in \Mod_R$ to $\unit_\cC$, and so takes $R^{\oplus \alpha}$ to $\unit_\cC^{\oplus \alpha}$.  Since $\unit_\cC$ is compact, $\hom_\cC(\unit_\cC,\unit_\cC^{\oplus \alpha}) = \End_\cC(\unit_\cC)^{\oplus \alpha}$ is $(\dim(\End_\cC(\unit_\cC)) \times \alpha)$-dimensional over $\bR$.  Since $\dim(\End_\cC(\unit_\cC)) < \infty$, the cardinal $\alpha$ is determined by the cardinal $\dim(\End_\cC(\unit_\cC)) \times \alpha$. This verifies the injectivity-on-objects condition in \cref{defn.categorified field}.
\end{proof}

We now answer the question of finding the categorified algebraic closure of $\bR$.  Recall that the symmetric monoidal category 
$\SuperVect_\bC$ of \define{supervector spaces} over $\bC$ is by definition equivalent as a monoidal category, but not as a symmetric monoidal category, to the category $\cat{Rep}_\bC(\bZ/2)$ of complex representations of the group $\bZ/2$.  Let   $\bJ$ denote the ``sign'' representation, also called the \define{odd line}.  In $\cat{Rep}_\bC(\mathbb Z/2)$, the symmetry $\bJ \otimes \bJ \to \bJ \otimes \bJ$ is multiplication $+1$; in $\SuperVect_\bC$ the symmetry is $-1$.  The rest of the symmetry is determined from this law by the axioms of a symmetric monoidal category.

The following is, with just a few changes of context, the main result of \cite{MR1944506}; because of these few changes, we review the proof.

\begin{theorem} \label{thm.Deligne}
  $\SuperVect_\bC$ is the unique (up to non-unique equivalence) finite-dimensional algebraically closed categorified field over $\bR$.
\end{theorem}

\begin{proof}
To show that $\SuperVect_\bC$ is a categorified field, one proceeds as in the proof of \cref{fields are categorified fields}.  The additional observation needed is the following.  Let $F : \SuperVect_\bC \to \cC$ be a morphism of finite-dimensional categorified commutative $\bR$-algebras, and let $\bJ_\cC = F(\bJ)$ denote the image of the odd line.  Then $\bJ_\cC$ has self-braiding $-1$ whereas $\unit_\cC$ has self-braiding $+1$, from which it follows that $\unit_\bC$ and $\bJ_\bC$ are not isomorphic.  On the other hand, tensoring with $\bJ_\cC$ induces an autoequivalence of $\cC$, and so $\bJ_\cC$, like $\unit_\cC$, is compact and non-zero.  From these facts, it follows that $F$ is faithful and that one can recover the isomorphism type of an object $V = \unit^{\oplus \alpha} \oplus \bJ^{\oplus \beta} \in \SuperVect_\bC$ from the vector space $\hom_\cC(\unit_\cC \oplus \bJ_\cC,F(V))$.

  We next verify that, assuming $\SuperVect_\bC$ is algebraically closed, it is the unique such category.  Suppose that $\cC$ is another algebraically closed finite-dimensional categorified field over $\bR$.  Then there are symmetric monoidal functors $\cC \to \SuperVect_\bC$ and $\SuperVect_\bC \to \cC$, both faithful and injective on objects.  Their composition $\SuperVect_\bC \to \cC \to \SuperVect_\bC$ is full and essentially surjective as it necessarily takes $\unit \mapsto \unit$ and $\bJ\mapsto\bJ$.  Thus the functor $\cC \to \SuperVect_\bC$ is essentially surjective and full (fullness uses that $\cC\to \SuperVect_\bC$ is injective on objects).

  Finally, we prove that $\SuperVect_\bC$ is algebraically closed.  
Let $\cC$ be a non-zero finite-dimensional categorified commutative $\bR$-algebra. We must construct a 1-morphism $\cC \to \SuperVect_\bC$.  By including $\cC \to \cC \boxtimes_\bR \SuperVect_\bC$ if necessary, we may assume without loss of generality that $\cC$  receives a 1-morphism $\SuperVect_\bC \to \cC$.  As above, we will denote the images under this 1-morphism of $\unit,\bJ \in \SuperVect_\bC$ by $\unit_\cC,\bJ_\cC$.

  We will need the following notion.
    Let $\lambda$ be a partition of $n \in \bN$ and $V_\lambda$ the corresponding irrep of the symmetric group $\bS_n$.  Recall that, for any $\bC$-linear symmetric monoidal category $(\cC,\otimes,\dots)$ containing direct sums and splittings of idempotents,  the \define{Schur functor} $S_\lambda : \cC \to \cC$ is the (nonlinear) functor $X \mapsto (X^{\otimes n}\otimes V_\lambda)_{\bS_n}$, where $\bS_n$ acts on $X^{\otimes n}$ via the symmetry on $\cC$, and $(-)_{\bS_n}$ denotes the functor of coinvariants.  
    $S_\lambda$ is natural for symmetric monoidal $\bC$-linear functors.

Choose a strong projective generator $P \in \cC$.  (In the notation of \cref{defn.finite dimensional}, one can for example take $P = \bigoplus_i C_i$.)  
Then the underlying category of $\cC$ is equivalent to the category of $\End_\cC(P)$-modules, and the subcategory of compact objects of $\cC$ is the abelian category of finite-dimensional $\End_\cC(P)$-modules.  In particular, every compact object has finite length.  
As shown in the proof of \cref{fields are categorified fields}, $\unit_\cC$ is compact, from which it follows that all dualizable objects are compact.  Since $P$ is dualizable by assumption, $P^{\otimes n}$ is also dualizable and hence compact.  

    We claim that there exists some $\lambda$ such that $S_\lambda(P) = 0$.  Indeed, suppose that there were not.  Then, as in \cite[Paragraph 1.20]{MR1944506}, the isomorphism $P^{\otimes n} \cong \bigoplus_{|\lambda| = n} V_\lambda \otimes S_\lambda(P)$, would imply that 
  $$\operatorname{length}(P^{\otimes n}) \geq \sum_{|\lambda| = n} \dim V_\lambda \geq \left( \sum (\dim V_\lambda)^2\right)^{1/2} = (n!)^{1/2},$$
  which grows more quickly than any geometric series.  Suppose that $X,Y,M \in \cC$ are compact objects and $E$ is an extension of $X$ by $Y$.  Then, as in \cite[Lemma 4.8]{MR1944506}, right exactness of the tensor functor implies:
  $$\operatorname{length}(E\otimes M) \leq \operatorname{length}(E\otimes X) + \operatorname{length}(E\otimes Y).$$  From this, the lengths of the tensor products of simple objects, and the fact that finite-dimensional algebras admit  only finitely many simple modules, one can bound the growth of $\operatorname{length}(P^{\otimes n})$ by some geometric series.

Given a commutative algebra object $A \in \cC$, let $\unit_A$ and $\bJ_A$ denote the images of $\unit_\cC$ and $\bJ_\cC$ under the extension-of-scalars functor $\otimes A : \cC \to \{\text{$A$-modules in $\cC$}\}$.  Note $\otimes_A$ makes $\{\text{$A$-modules in $\cC$}\}$ into a categorified commutative $\bR$-algebra.   Following \cite[Proposition 2.9]{MR1944506}, we will find a non-zero commutative algebra $A \in \cC$ such that $P \otimes A \cong \unit_A^{\oplus r} \oplus \bJ_A^{\oplus s}$ for some $r,s\in \bN$. 
Supposing we have done so, let $R$ be the commutative superalgebra whose even part is $\End(\unit_A)\cong\End(\bJ_A)$ and whose odd part is $\hom(\bJ_A,\unit_A)\cong \hom(\unit_A,\bJ_A)$, i.e.\ the ``endomorphism superalgebra'' of $\unit_A$.  Since $P$ is a compact projective generator of $\cC$ and $P \otimes A \cong \unit_A^{\oplus r} \oplus \bJ_A^{\oplus s}$, the symmetric monoidal category $\{\text{$A$-modules in $\cC$}\}$ is strongly generated as a category by $\unit_A$ and $\bJ_A$, and so is equivalent to the category $\cat{SuperMod}_R$ of $R$-modules in $\SuperVect_\bC$; this equivalence is then manifestly symmetric monoidal.

   Suppose by induction that we have found a non-zero commutative algebra object $A \in \cC$ such that $P \otimes A \cong \unit_A^{\oplus r'} \oplus \bJ_A^{\oplus s'} \oplus P'$ for some $P' \in \{\text{$A$-modules in $\cC$}\}$.  Then $P'$ is a summand of a dualizable object and hence dualizable.  If $\Sym^n P' = \bigwedge^n P' = 0$ for all sufficiently large $n$, then $P' = 0$ by \cite[Corollary 1.7 and Lemma 1.17]{MR1944506}.  If on the other hand $\Sym^n P' \neq 0$ for all $n$ (resp.\ $\bigwedge^n P' \neq 0$ for all $n$), then \cite[Lemma 2.8]{MR1944506}, which does not assume the category to be rigid, constructs a non-zero $A$-algebra $A'$ such that $P' \otimes A' \cong \unit_{A'} \oplus P''$ (resp.\ $P' \otimes A' \cong \bJ_{A'} \oplus P''$).  We iterate, continually splitting off $\unit_A$s and $\bJ_A$s.  The iteration must terminate as otherwise $S_\lambda(P) \neq 0$ for all $\lambda$ \cite[Corollary 1.9]{MR1944506}.

   Thus we have found a non-zero commutative superalgebra $R$ and a morphism $\cC \to \cat{SuperMod}_R$  of categorified commutative $\bR$-algebras.  We can choose a field $\bL$ that receives a map from $R$ and extend scalars further so as to build a linear cocontinuous symmetric monoidal functor $\cC \to \SuperVect_\bL$.    Moreover, since $\End_\cC(P)$ is finite-dimensional over $\bC$, the functor $\cC \to \SuperVect_\bL$ factors through $\SuperVect_{\bK}$ for some intermediate field $\bC \subseteq \bK \subseteq \bL$ which is finite-dimensional over $\bC$.  But since $\bC$ is algebraically closed, the only such field is $\bK = \bC$.
\end{proof}

\begin{remark}
  The fact that $\SuperVect_\bC$ is algebraically closed explains its central role in ``categorified'' representation theory \cite{MR2165457,MR3177367}.
\end{remark}

\begin{remark}\label{remark.char p version}
  The categorified algebraic closure of $\bar{\bF}_p$ is not yet known.  When $p>2$, \cite{Ostrik2015} conjectures that the answer is a characteristic-$p$ version of ``quantum $\mathrm{SU}(2)$ at level $p-2$'' called $\cat{Ver}_p$.  Etingof has conjectured that the categorified algebraic closure of $\bar{\bF}_2$ is a non-semisimple characteristic-$2$ version of $\SuperVect$, described by the triangular Hopf algebra $\bar{\bF}_2[x]/(x^2)$ with $\Delta(x) = 1\otimes x + x\otimes 1$ and $R$-matrix $R = 1\otimes 1 + x\otimes x$.
\end{remark}

We now use the algebraic closure $\Vect_\bR \to \SuperVect_\bC$ to categorify the notion of ``torsor over $\Spec(\bR)$.'' We first show that $\Vect_\bR \to \SuperVect_\bC$ is ``Galois.''  Let $(\cC,\otimes_\cC,\dots)$ be a categorified commutative $\bR$-algebra.  A \define{$\cC$-module} is an $\bR$-linear locally presentable category $\cV \in \Pres_\bR$ together with an action of $\cC$ on $\cV$ which is cocontinuous in each variable.  A morphism of finite-dimensional $\cC$-modules is a cocontinuous strong module functor.  Since $\cC$ is commutative, the bicategory $\MOD_\cC$ of finite-dimensional $\cC$-modules carries a symmetric monoidal structure $\boxtimes_\cC$.  See for example Definitions 2.1, 2.6, and 3.2 of \cite{DSPS2}.

  Let $(\cC,\otimes_\cC,\dots) \to (\cD,\otimes_\cD,\dots)$ be a 1-morphism of categorified commutative $\bR$-algebras.  Such a map makes $(\cD,\otimes_\cD,\dots)$ into a commutative algebra object in $\MOD_\cC$.  Let $\Aut = \Aut_\cC(\cD)$ denote the group of $\cC$-linear symmetric monoidal automorphisms of $\cD$.  
We will denote by $\MOD_{\cD \rtimes \Aut}$  the bicategory of $\cD$-modules equipped with a $\cC$-linear $\Aut$-action such that the $\cD$-action is $\Aut$-equivariant.  It is a symmetric monoidal bicategory with symmetric monoidal structure given by the tensor product of underlying $\cD$-modules.  The scalar extension functor $\boxtimes_\cC \cD : \MOD_\cC \to \MOD_\cD$ factors canonically through $\MOD_{\cD \rtimes \Aut}$:
$$ \MOD_\cC \overset{\boxtimes_\cC \cD}\longrightarrow \MOD_{\cD \rtimes \Aut} \overset{\text{forget}}\longrightarrow \MOD_\cD$$
The functor $\boxtimes_\cC \cD : \MOD_\cC \to \MOD_{\cD \rtimes \Aut}$ has a  right adjoint $(-)^{\Aut} :  \MOD_{\cD \rtimes \Aut} \to  \MOD_\cC$ given by taking the $\Aut$-fixed points of a module $\cV \in  \MOD_{\cD \rtimes \Aut}$.

\begin{definition}
  An extension of categorified fields $(\cC,\otimes_\cC,\dots) \to (\cD,\otimes_\cD,\dots)$ is \define{Galois} if 
  $$ \boxtimes_\cC \cD :\MOD_\cC \leftrightarrows \MOD_{\cD \rtimes \Aut} :(-)^{\Aut}$$ is an equivalence of bicategories.
\end{definition}

We will prove:
\begin{theorem} \label{supervect is Galois}
  The extension $\Vect_\bR \to \SuperVect_\bC$ is Galois.
\end{theorem}

\begin{remark}
  For comparison, the extensions $\Vect_{\bar{\bF}_p} \to \cat{Ver}_p$ and $\Vect_{\bar{\bF}_2} \to \cat{Rep}(\bF[x]/(x^2))$ from \cref{remark.char p version} are not Galois (except for when $p=3)$.  Indeed, the latter is ``purely inseparable,'' and the maximal ``separable'' subextension of $\Vect_{\bar{\bF}_p} \to \cat{Ver}_p$ is $\SuperVect_{\bar{\bF}_p}$.
\end{remark}

We henceforth write $\cat{Gal}(\bR) = \Aut_\bR(\SuperVect_\bC)$, and call it the \define{categorified absolute Galois group of $\bR$}.  We first calculate it:
\begin{lemma} \label{categorified abs Galois group of R}
  The categorified absolute Galois group of $\bR$ is $\bZ/2 \times \rB(\bZ/2)$.
\end{lemma}

\begin{proof}
  Since categorified commutative $\bR$-algebras form a bicategory,  $\Aut_{\bR}(\SuperVect_\bC)$ is a group object in 
  homotopy 1-types.
   A symmetric monoidal autoequivalence of $\SuperVect_\bC$ consists of a functor $F : \SuperVect_\bC \to \SuperVect_\bC$ and  some compatible isomorphisms.  We can canonically trivialize the isomorphisms $F(\unit) \cong \unit$ and $F(\unit \otimes X) \cong \unit \otimes F(X)$, and so the only remaining datum is an isomorphism $\phi: F(\bJ \otimes \bJ) \isom F(\bJ) \otimes F(\bJ)$, of which there are $\bC^\times$-many.  
  The functor $F$ admits symmetric monoidal natural automorphisms that are trivial on $\unit$ but act on $\bJ$ by $\alpha \in \bC^\times$.  Under such an automorphism, the map $\phi$ transforms to $\phi\alpha^2$.  Thus we find that $\Aut(F) \cong \operatorname{ker} \bigl( \bC^\times \overset{\alpha \mapsto \alpha^2}\longrightarrow \bC^\times \bigr) \cong \bZ/2$.
    
   $F$ induces an automorphism of $\bC = \End(\unit)$.  If this is the identity, then $F$ is  monoidally equivalent to the identity; otherwise, $F$ is monoidally equivalent to extension of scalars along the complex conjugation map $\bC \to \bC$.  Thus $\pi_0(\Aut_{\bR}(\SuperVect_\bC)) \cong \bZ/2$, and the above computation shows that each connected component is a $\rB(\bZ/2)$.  These fit together via the Galois action of $\bZ/2$ on $\operatorname{ker} \bigl( \bC^\times \overset{\alpha \mapsto \alpha^2}\longrightarrow \bC^\times \bigr)$, and so $\Aut_{\bR}(\SuperVect_\bC)$ is a split extension $\bZ/2 \ltimes \rB(\bZ/2)$.  Direct calculation verifies that it is the trivial extension; one can also show via standard techniques that there are no nontrivial split extensions of $\bZ/2$ by $\rB(\bZ/2)$.
\end{proof}

  The nontrivial element in $\pi_1\rB(\bZ/2)$ acts on $\SuperVect_\bC$ as the natural transformation of the identity commonly called ``$(-1)^f$,'' where $f$ stands for ``fermion number.''

\begin{proof}[Proof of \cref{supervect is Galois}]
  The bicategory $\MOD_{\Vect_\bR}$ is nothing but $\Pres_\bR$ itself.  Given $\cV \in \Pres_\bR$, its image under $\boxtimes_\bR \SuperVect_\bC$ in $\MOD_{\SuperVect_\bC\rtimes \Aut}$ can be described as follows.  The objects of $\cV \boxtimes_\bR \SuperVect_\bC$ are formal direct sums $V_0 \oplus \bJ V_1$ where $V_0$ and $V_1$ are objects of $\cV$.  The morphisms are $\hom(V_0\oplus\bJ V_1,W_0 \oplus \bJ W_1) = \hom_\cV(V_0,W_0) \otimes \bC \oplus \hom_\cV(V_1,W_1) \otimes \bC$.  $\SuperVect_\bC$ acts on $\cV \boxtimes_\bR \SuperVect_\bC$ in the obvious way.  The action of $\Aut_{\bR}(\SuperVect_\bC) = \bZ/2 \times \rB(\bZ/2)$ is via complex conjugation and $(-1)^f$, just as it is on $\SuperVect_\bC$.  The fixed-points of this action are therefore the ``purely even'' objects --- those of the form $V_0 \oplus \bJ 0$ --- equipped with a $\bC$-antilinear involutive automorphism of $V_0$.  
  The fact that $\bR \to \bC$ is Galois then implies that the composition $(-)^{\Aut} \circ (\boxtimes_\bR \SuperVect_\bC)$ is equivalent to the identity.

  It remains to verify that $(\boxtimes_\bR \SuperVect_\bC) \circ (-)^{\Aut}$ is equivalent to the identity.  Let $\cV$ be a $\SuperVect_\bC$-module.  Then $\cV$ comes equipped with an endofunctor $\bJ \otimes : \cV \to \cV$, given by the action of the odd line $\bJ \in \SuperVect_\bC$, satisfying $(\bJ \otimes)^2 \cong \id$, and for each $X \in \Vect_\bC$ an endofunctor $X \otimes : \cV \to \cV$.  The data of an $\Aut$-action on $\cV$ compatible with these actions consists of: an endofunctor $V \mapsto \bar V$, squaring to the identity, such that for $X \in \Vect_\bC$, $\overline{X\otimes V} \cong \bar X \otimes \bar V$; and a natural automorphism $\theta$ of the identity functor, squaring to the identity, such that $\theta_{\bJ \otimes V} = -\id_\bJ \otimes \theta_V$.  Let's say that $V \in \cV$ is \define{purely even} if $\theta_V = +1$ and \define{purely odd} if $\theta_V = -1$.  Then every $V\in \cV$ canonically decomposes into a direct sum $V = V_0 \oplus V_1$ of purely even and purely odd submodules.  The $\Aut$-fixed points are the purely even submodules $V = V_0$ equipped with isomorphisms $V \cong \bar V$.  
  Note that $\bJ \otimes$ interchanges purely even and purely odd objects, and so
   \begin{align*} \cV & \simeq \{\text{purely even objects in }\cV\} \boxplus \{\text{purely odd objects in }\cV\} \\ & \simeq \{\text{purely even objects in }\cV\} \boxtimes_\bC \SuperVect_\bC.\end{align*}
   Finally, since $\bR \to \bC$ is Galois, restricting from $\{\text{purely even objects in }\cV\}$ to those with $V \cong \bar V$ gives an $\bR$-linear category which tensors with $\bC$ to $\{\text{purely even objects in }\cV\}$.
\end{proof}

\begin{remark}\label{supervectH}
  \Cref{supervect is Galois} implies that the full list of categorifield field extensions of $\bR$ consists of the familiar categories $\Vect_\bR$, $\Vect_\bC$, $\SuperVect_\bR$, and $\SuperVect_\bC$, and a less-familiar category that deserves to be called $\SuperVect_\bH$.  The first four are the fixed-points for the obvious subgroups $\bZ/2 \times \rB(\bZ/2)$, $\rB(\bZ/2)$, $\bZ/2$, and $\{1\}$ of the categorified Galois group $\cat{Gal}(\bR)$.  The last is the fixed-points for the non-obvious inclusion $\bZ/2 \hookrightarrow \bZ/2 \times \rB(\bZ/2)$ which is the identity on the first component and the nontrivial map $\bZ/2 \to \rB(\bZ/2)$ on the second component (corresponding to the nontrivial class in $\H^2(\rB(\bZ/2);\bZ/2)$).  As a category, $\SuperVect_\bH \simeq \Vect_\bR \boxplus \Mod_\bH$, hence the name.  The monoidal structure involves the Morita equivalence $\bH \otimes_\bR \bH \simeq \bR$.
\end{remark}

We can now categorify the usual classification of torsors in terms of Galois actions.

\begin{definition} \label{defn.torsor}
  Let $G$ be a finite Picard groupoid.  A \define{categorified $G$-torsor over $\bR$} is a non-zero $G$-equivariant categorified commutative $\bR$-algebra $\cT$ such that 
  the functor
  $$ \cT \boxtimes_\bR \cT \to \maps(G,\cT), \quad t_1 \boxtimes t_2 \mapsto \bigl( g \mapsto (g\triangleright t_1) \otimes t_2 \bigr) $$
  is an equivalence, where $\maps(G,\cT)$ denotes the categorified commutative algebra of $\cT$-valued functors on the underlying groupoid of $G$,  $\triangleright$ denotes the action of $G$ on $\cT$, and $\otimes$ denotes the multiplication in $\cT$.
\end{definition}

\begin{proposition}\label{categorified torsor statement}
  Let $\cat{Gal}(\bR) = \Aut_\bR(\SuperVect_\bC)$ denote the {categorified absolute Galois group} of $\bR$.  For each finite categorified group $G$, there is a natural-in-$G$ equivalence
  $$ \{\text{categorified $G$-torsors over $\bR$}\} \simeq \maps(\rB\operatorname{\cat{Gal}}(\bR),\rB G).$$
\end{proposition}

The proof is just as in the uncategorified situation:

\begin{proof}
  Let $\cT$ be a categorified $G$-torsor over $\bR$.  Then $\cT' = \cT \boxtimes_\bR \SuperVect_\bC$ is a $G$-torsor over $\SuperVect_\bC$. 
  Since $\SuperVect_\bC$ is algebraically closed, we can choose a symmetric monoidal functor $F : \cT' \to \SuperVect_\bC$.  
  Let $\boxtimes'$ denote $\boxtimes_{\SuperVect_\bC}$.  The equivalence $\cT' \boxtimes' \cT' \mapsto \maps(G,\cT')$ making $\cT'$ into a torsor over $\SuperVect_\bC$ fits into a commutative square
   $$ \begin{tikzpicture}[auto]
   \path node (UL) {$\cT' \boxtimes' \cT'$} +(5,0) node (UR) {$\maps(G,\cT')$} +(0,-2) node (LL) {$\cT'$} +(5,-2) node (LR) {$\maps(G,\SuperVect_\bC)$}
     ;
   \draw[->] (UL) -- node {$\sim$} (UR); \draw[->] (UL) -- node {$\scriptstyle \id \boxtimes' F$} (LL); \draw[->] (LL) -- node {$\scriptstyle t \mapsto (g\mapsto F(g\triangleright t))$} (LR); \draw[->] (UR) -- node {$\scriptstyle F \circ$} (LR);
  \end{tikzpicture}$$
  in which the downward arrows are both equivalent to $\boxtimes_{\cT'} \SuperVect_\bC$.  It follows that $\cT'$ is a trivial $G$-torsor over $\SuperVect_\bC$.  
  
   \Cref{supervect is Galois} then provides an equivalence of homotopy 2-types $\{$categorified $G$-torsors over $\bR\} \simeq \{\cat{Gal}(\bR)$-actions on the trivial $G$-torsor over $\SuperVect_\bC$ compatible with the action on $\SuperVect_\bC\}$.
  But $\Aut_\bR(\maps(G,\SuperVect_\bC)) \simeq G \times \cat{Gal}(\bR)$,
    and the equivariance requirement is equivalent to the requirement that the morphism $\cat{Gal}(\bR) \to G \times \cat{Gal}(\bR)$ is the identity on the second component.  Therefore we are left with maps $\cat{Gal}(\bR) \to G$ up to equivalences given by inner automorphism.
\end{proof}

\section{Spin and Spin-Statistics field theories} \label{section.spin statistics}

With the categorified Galois extension $\Vect_\bR \to \SuperVect_\bC$ from \cref{section.algebraic closure} in hand, we are equipped to categorify the story from \cref{section.hermitian}.  The uncategorified story related orientations with Hermiticity; the categorified story will relate spin and statistics.

Recall that a \define{spin structure} on a $d$-dimensional manifold $M$ is a $\Spin(d)$-principal bundle $P \to M$ together with an isomorphism $P \times_{\Spin(d)}\bR^d \cong \rT M$.  The collection $\Spins(M)$ of spin structures on $M$ is not naturally a set, but rather a groupoid.  We therefore extend without further comment the notion of \define{topological local structure} valued in a bicategory $\sX$ to be a sheaf $\Man_d \to \sX$ that takes homotopies between maps in $\Man_d$ to isomorphisms between maps in $\sX$ and homotopies between homotopies to equalities between isomorphisms.
Generalizing \cref{set theory cob hyp}, we have:
\begin{lemma} \label{gpoid cob hyp}
Let $\sX$ be a bicategory with limits.
  Topological local structures on $\Man_d$ valued in $\sX$ are equivalent to objects of $\sX$ equipped with an action by the Picard groupoid $\pi_{\leq 1}\hom_{\Man_d}(\bR^d,\bR^d) = \pi_{\leq 1}\rO(d)$.  When $d\geq 3$, this Picard groupoid is canonically equivalent to $\bZ/2 \times \rB(\bZ/2)$. \qedhere
\end{lemma}

\begin{remark}\label{canonical splitting}
  The existence of an identification $\pi_{\leq 1}\rO(d) \cong \bZ/2 \times \rB(\bZ/2)$, $d\geq 3$, is the same as the standard assertion that the k-invariant connecting $\pi_1\rB\rO(\infty)$ and $\pi_2 \rB\rO(\infty)$ vanishes.    
  However, the group $\bZ/2 \times \rB(\bZ/2)$ admits a nontrivial group automorphism, given by the identity on each factor and the nontrivial group map $\bZ/2 \to \rB(\bZ/2)$, corresponding to the nontrivial element of $\rH^2(\rB(\bZ/2);\bZ/2) = \bZ/2$, mixing the factors.  
  Thus there are two inequivalent identifications $\pi_{\leq 1}\rO(\infty) \cong \bZ/2 \times \rB(\bZ/2)$.
  To pick one is the same as to pick a splitting of the projection $\pi_{\leq 1}\rO(d) \to \pi_0\rO(d) = \bZ/2$.  There is a canonical choice: the ``stable'' splitting $\bZ/2 \to \rO(d)$ sending the nontrivial element of $\bZ/2$ to $\left( \begin{smallmatrix} -1 & && \\ &+1&& \\ && +1 & \\ &&& \ddots \end{smallmatrix}\right)$, called ``T'' in the physics literature.
  Corresponding to the two splittings $\bZ/2 \to \pi_{\leq 1}\rO(d)$ are two projections $\pi_{\leq 1}\rO(d) \to \rB(\bZ/2)$, the kernels of which  are the two \define{pin} groups $\Pin^\pm(d)$.
\end{remark}

\begin{example}\label{eg.spins locally}
  We recall two standard facts about spin structures.  First,
given any spin structure on a $d$-dimensional manifold $M$, let $P \to M$ denote the corresponding $\Spin(d)$-bundle.  Then $P \times_{\Spin(d)} \Pin^+(d)$ is a $\Pin^+(d)$-bundle over $M$ with a distinguished sheet.  The other sheet of $P \times_{\Spin(d)} \Pin^+(d)$ also defines a $\Spin(d)$-bundle over $M$, corresponding to the \define{orientation reversal} of the original spin structure.  
Second, any spin structure on $M$ admits a square-$1$ automorphism which acts on the bundle $P \to M$ by multiplication by the nontrivial central element of $\Spin(d)$ coming from $360^\circ$-rotation in $\SO(d)$.  The mapping cylinder of this automorphism is the product spin  manifold $M \times \twist$, where $\twist$ denotes the nontrivial-rel-boundary spin structure on the interval $[0,1]$.  We will also use the name ``$\twist$'' to denote the automorphism of the spin structure.  Together, orientation reversal and $\twist$ define an action of $\bZ/2 \times \rB(\bZ/2)$ on $\Spins(M)$.

When $M = \bR^d$, the orientation reversal and $360^\circ$-rotation action of $\bZ/2 \times \rB(\bZ/2)$ witness $\Spins(\bR^d)$ as the trivial $(\bZ/2 \times \rB(\bZ/2))$-torsor.  When $d\geq 3$, orientation reversal and $360^\circ$-rotation comprise the full group $\pi_{\leq 1} \rO(d) = \bZ/2 \times \rB(\bZ/2)$, and so $\Spins$ is the topological local structure corresponding to the trivial $\pi_{\leq 1} \rO(d)$-torsor via \cref{gpoid cob hyp}.  When $d< 3$, the canonical inclusion $X \mapsto \bigl( \begin{smallmatrix} X & \\ & 1 \end{smallmatrix}\bigr)$ of $\rO(d)$ into $\rO(3)$ provides an action of $\pi_{\leq 1}\rO(d)$ on $\pi_{\leq 1}\rO(3) \cong \bZ/2 \times \rB(\bZ/2)$, which in turn corresponds to the topological local structure $\Spins$.
\end{example}

We now move to an algebrogeometric setting in which there are interesting topological local structures that are \'etale-locally equivalent to $\Spins$ in the way that $\Her$ was \'etale-locally equivalent to $\Or$.  Ordinary algebraic geometry does not suffice, since $\Spec(\bC)$ is \'etale-contractible in the ordinary sense.  Instead, since groupoids are a categorification of sets, we work with a categorification of schemes:
\begin{definition}\label{defn.categorified stack}
  The bicategory $\CatAffSch_\bR$ of \define{categorified affine schemes over $\bR$} is opposite to the bicategory of categorified commutative $\bR$-algebras in the sense of \cref{defn.cat com alg}.  We will write $\Spec(\cC)$ for the categorified affine scheme corresponding to a categorified commutative algebra $\cC$.
\end{definition}

  \Cref{lemma.every affine is categorifiable} provides a fully faithful inclusion of the category $\cat{AffSch}_\bR$ of uncategorified affine schemes into $\CatAffSch_\bR$; in particular, we identify $\Spec(\bR)$ with $\Spec(\Vect_\bR)$.  The details of notions like ``non-affine categorified scheme'' and ``categorified \'etale topology'' have yet to be worked out, and are the subject of joint work in progress by A.\ Chirvasitu, E.\ Elmanto, and the author.
\Cref{supervect is Galois} suggest that $\Spec(\SuperVect_\bC) \to \Spec(\bR)$ is a ``categorified \'etale cover'' and \cref{thm.Deligne} suggests that $\Spec(\SuperVect_\bC)$ is ``categorified \'etale contractible.''  In particular, we will say that categorified affine schemes $X$ and $Y$ are \define{\'etale-locally equivalent} if their pullbacks $X \times_{\Spec(\bR)} \Spec(\SuperVect_\bC)$ and $Y \times_{\Spec(\bR)}\Spec(\SuperVect_\bC)$ are equivalent as categorified affine schemes over $\SuperVect_\bC$.  This in particular implies that for any Picard groupoid $G$, the geometric notion of ``categorified $G$-torsors over $\Spec(\bR)$,'' defined as $G$-objects over $\Spec(\bR)$ \'etale-locally equivalent to $G$ acting on itself, agrees with the algebraic notion from \cref{defn.torsor}, which by \cref{categorified abs Galois group of R,categorified torsor statement} are classified by maps $\bZ/2 \times \rB(\bZ/2) \to G$.  

Now note the following coincidence: there is a canonical equivalence $\pi_{\leq 1}\rO(d) \cong \bZ/2 \times \rB(\bZ/2) \cong \cat{Gal}(\bR)$, and hence a canonical categorified $\pi_{\leq 1}\rO(d)$-torsor, when $d\geq 3$.  This torsor is nothing but the categorified affine scheme $\Spec(\SuperVect_\bC)$ equipped with its $\cat{Gal}(\bR)$-action.

\begin{definition}\label{defn.spinstats}
  The sheaf of \define{Hermitian spin-statistics structures} is the sheaf $\SpinStats: \Man_d \to \CatAffSch_\bR$ such that $\SpinStats(\bR^d) = \Spec(\SuperVect_\bC)$ on which $\pi_{\leq 1}\rO(d) \cong \cat{Gal}(\bR)$ acts via the Galois action.
\end{definition}

\begin{lemma} \label{spins versus spinstats}
  For any manifold $M$, 
  $$ \SpinStats(M) = \frac{\Spins(M) \times \Spec(\SuperVect_\bC)}{\bZ/2 \times \rB(\bZ/2)}, $$
  where $\bZ/2 \times \rB(\bZ/2)$ acts on $\Spins(M)$ by orientation reversal and $\twist$ from \cref{eg.spins locally}, and it acts on $\SuperVect_\bC$ by complex conjugation and $(-1)^f$ from \cref{categorified abs Galois group of R}. \qedhere
\end{lemma}

\Cref{spins versus spinstats} begins to justify the name ``Hermitian spin-statistics structure'' in \cref{defn.spinstats}: that orientation reversal acts by complex conjugation is the essence of Hermiticity, and that $\twist$ acts by $(-1)^f$ is a version of ``spin-statistics'' as it is used in physics. 

 To further justify the name, we should study Hermitian spin-statistics field theories directly.  The definition of ``Heritian spin-statistics field theory'' will be a direct analog of ``Hermitian field theory'' from \cref{section.hermitian}.  

Let $\Bord_{d-2,d-1,d}$ denote the ``once-extended'' $d$-dimensional bordism bicategory constructed in \cite{Schommer-Pries:thesis}.  Given any topological local structure valued in groupoids $\cG : \Man_d \to \cat{Gpoids}$, \cite{Schommer-Pries:thesis} also explains how to build a symmetric monoidal bicategory $\Bord_{d-2,d-1,d}^\cG$ of bordisms with $\cG$-structure.  A \define{once-extended $\cG$-structured field theory} is then a symmetric monoidal functor $Z : \Bord_{d-2,d-1,d}^\cG \to \sV$ for some symmetric monoidal bicategory $\sV$ of ``categorified vector spaces.''  

We will take $\sV = \Alg_\bR$ to be the symmetric monoidal ``Morita'' bicategory of associative algebras, bimodules, and intertwiners.  Just as $\Vect_\bR$ had a natural extension to the stack $\Qcoh$ of categories over $\Sch_\bR$, so too $\Alg_\bR$ has a natural extension allowing for ``bundles'' or ``sheaves'' of algebras over any categorified affine scheme: given a categorified commutative $\bR$-algebra $\cC$, set $\Alg(\Spec(\cC)) = \Alg(\cC)$ to be the symmetric monoidal bicategory of algebra objects in $\cC$, bimodule objects in $\cC$, and intertwiners in $\cC$.  Although we have not defined, and will not use, any topology on $\CatAffSch_\bR$, and so cannot say precisely what it means to be a ``stack of bicategories,'' it is not hard to find a bicategory object internal to $\CatAffSch_\bR$ that represents $\Alg(-)$, and so $\Alg(-)$ is certainly  a stack of bicategories in any subcanonical topology.

\begin{remark}
  The Eilenberg--Watts theorem \cite{MR0125148,MR0118757} identifies $\Alg_\bR$ with the full subbicategory of $\Pres_\bR$ whose objects admit a compact projective generator.  The ``correct'' target for once-extended \emph{non-topological} quantum field theory is more likely the larger $\Pres_\bR$.  But it is reasonable to expect that every \emph{topological} field theory factors through $\Alg_\bR$, since it is expected that only categories equivalent to $\Mod_A$, $A \in \Alg_\bR$, are sufficiently ``dualizable''  (c.f.\ \cite{BCJF2014}).  Indeed, one should expect more: topological  field theories should factor through the subbicategory of $\Alg_\bR$ whose objects are finite-dimensional algebras and whose morphisms are finite-dimensional bimodules.  This subbicategory is equivalent to the bicategory $\MOD_{\Pres_\bR}$ of finite-dimensional $\Pres_\bR$-modules from \cref{section.algebraic closure}.  More generally, for $\cC$ a finite-dimensional categorified commutative ring, the bicategory $\MOD_\cC$ of finite-dimensional $\cC$-modules is a subbicategory of $\Alg(\cC)$, which is a subbicategory of the bicategory of all $\cC$-modules.
\end{remark}

\begin{definition} \label{bicategorical description of G-structures}
  Let $\Spans_2(\CatAffSch_\bR)$ denote the symmetric monoidal bicategory whose objects are categorified affine schemes, 1-morphisms are spans $X \leftarrow A \rightarrow Y$, and 2-morphisms are spans-between-spans:
  $$
  \begin{tikzpicture}
    \path (0,0) node (A) {$X$} (1.5,0) node (Z) {$M$} (3,0) node (B) {$Y$}  (1.5,1) node (X) {$A$} (1.5,-1) node (Y) {$B$};
    \draw[->] (X) -- (A); \draw[->] (X) -- (B);
    \draw[->] (Y) -- (A); \draw[->] (Y) -- (B);
    \draw[->] (Z) -- (X); \draw[->] (Z) -- (Y);
  \end{tikzpicture}
  $$
Composition is by fibered product, and the symmetric monoidal structure is the cartesian product in $\CatAffSch_\bR$.  Let $\cG$ be a topological local structure valued in $\CatAffSch_\bR$; it defines a symmetric monoidal functor $\widetilde\cG : \Bord_{d-2,d-1,d} \to \Spans_2(\CatAffSch_\bR)$.  Let $\Spans_2(\CatAffSch_\bR;\Alg)$ be the symmetric monoidal bicategory whose objects are a categorified affine scheme $X$ together with an algebra $V \in \Alg(X)$, whose 1-morphisms are spans $X \overset f \leftarrow A \overset g \to Y$ together with a bimodule between $f^*V$ and $g^*W$ in $\Alg(A)$, and whose 2-morphisms are spans of spans together with an intertwiner between pulled-back bimodules.  A \define{$\cG$-structured field theory} is  a lift:
$$ \begin{tikzpicture}[baseline=(LR.base)]
   \path coordinate (UL) +(3,0) node (UR) {$\Spans_2(\CatAffSch_\bR;\Alg)$} +(-2,-1.5) node (LL) {$\Bord_{d-2,d-1,d}$} +(3,-1.5) node (LR) {$\Spans_2(\CatAffSch_\bR)$}
    ;
   \draw[->] (LL) -- node[auto] {$\scriptstyle \widetilde\cG$} (LR); \draw[->] (UR) -- node[auto] {\scriptsize {Forget} the $\Alg$-data} (LR);
   \draw[->,dashed] (LL) -- (UR);
  \end{tikzpicture}$$
\end{definition}

\begin{example}
  We now continue to justify the name ``Hermitian spin-statistics'' from \cref{spins versus spinstats}.  Let $Z$ be a $d$-dimensional $\SpinStats$-structured field theory.  We will unpack its values on various manifolds.
  
  Suppose first that $M$ is a closed $d$-dimensional manifold.  As an element of $\Bord_{d-2,d-1,d}$, $M$ is an endo-2-morphism of the identity 1-morphism of the unit object.  Then $Z(M)$ is  an endo-2-morphism of the identity 1-morphism of the unit object in $\Alg(\SpinStats(M))$, i.e.\ a function $Z(M) \in \cO(\SpinStats(M))$.  Any choice of spin structure for $M$ determines a map $\Spec(\SuperVect_\bC) \to \SpinStats(M)$, and these maps together cover $\SpinStats(M)$ as the spin structure varies over $M$.  Thus the data of $Z(M)$ is the data of an element of $\cO(\Spec(\SuperVect_\bC)) = \bC$ for each spin structure on $M$.  By the construction of $\SpinStats$ from \cref{spins versus spinstats}, two spin structures on $M$ with reversed orientation lead to complex-conjugate values of $Z(M)$.  This is a manifestation of the Hermiticity of $Z$.
  
  To see  spin-statistics phenomena, consider next the case of $N$ a closed $(d-1)$-dimensional manifold.  Then $Z(N)$ is an endo-1-morphisms of the unit object in $\Alg(\SpinStats(N))$, i.e.\ a vector bundle on $\SpinStats(N)$.  Again any spin structure on $N$ allows this vector bundle to be pulled back to a vector bundle on $\Spec(\SuperVect_\bC)$, and so $Z(N)$ assigns a complex supervector space to each spin structure on $N$.  In addition to the Hermiticity requirement that orientation-versed spin structures map to complex-conjugate supervector spaces, there is another relation between these supervector spaces and the spin structures.  Indeed, fix a spin structure $\sigma$ on $N$, and let $Z(N,\sigma)$ denote the corresponding complex supervector space.  Consider the spin cobordism $(N,\sigma) \times \twist$.  This spin structure picks out a particular map $\Spec(\SuperVect_\bC) \to \SpinStats(N \times[0,1])$, along which $Z(N \times [0,1])$ pulls back to a map $Z((N,\sigma) \times \twist) : Z(N,\sigma) \to Z(N,\sigma)$.  But $(N,\sigma) \times \twist$ is simply the mapping cylinder of the $360^\circ$-rotation of $\sigma$, and \cref{spins versus spinstats} identifies $360^\circ$-rotation with $(-1)^f$.  All together, we find that $Z((N,\sigma) \times \twist)$ is required to evaluate to $(-1)^f : Z(N,\sigma) \to Z(N,\sigma)$.
  
  Similar discussion applies also in codimension-2, and  Hermitian spin-statistics field theories unpack to spin field theories $\Bord_{d-2,d-1,d}^{\Spins}\to \Alg(\SuperVect_\bC)$ such that the actions of $\bZ/2 \times \rB(\bZ/2)$ on the source and target categories are intertwined.  The phrase ``spin-statistics'' refers to the identification $\twist = (-1)^f$.  In a spin field theory the $(-1)$-eigenstates of $\twist$ are called \define{spinors} and in a super field theory the $(-1)$-eigenstates of $(-1)^f$ are called \define{fermions}, so ``{spin-statistics}'' can be equivalently described as the assertion that the classes of spinors and fermions agree.
\end{example}

By construction, $\SpinStats$ is an example of a \define{\'etale-locally-spin} topological local structure in the sense that $\SpinStats \times_{\Spec(\bR)} \Spec(\SuperVect_\bC)$ and $\Spins \times \Spec(\SuperVect_\bC)$ are equivalent.  Since $\cat{Gal}(\bR) = \bZ/2 \times \rB(\bZ/2)$ and $\Spins$ corresponds to the trivial $\bZ/2 \times \rB(\bZ/2)$-torsor, \cref{categorified torsor statement} asserts that the set of inequivalent topological local structures \'etale-locally-equivalent to $\Spins$ is equivalent to $\pi_0\maps(\rB(\bZ/2 \times \rB(\bZ/2)), \rB(\bZ/2 \times \rB(\bZ/2)))$, which can be easily computed as
$$ \rH^1(\rB(\bZ/2);\bZ/2) \times \rH^2(\rB(\bZ/2);\bZ/2) \times \rH^1(\rB^2(\bZ/2);\bZ/2) \times \rH^2(\rB^2(\bZ/2);\bZ/2) \cong (\bZ/2)^3, $$
and so there are exactly eight different choices.  Whether the corresponding field theories are oriented or Hermitian is controlled by the component $\bZ/2 \to \bZ/2$ relating complex conjugation with orientation reversal.  Whether the field theories are spin or spin-statistics is controlled by the component $\rB(\bZ/2) \to \rB(\bZ/2)$ relating $(-1)^f$ with $\twist$.  But once these choices are made, there is still the choice of map $\bZ/2 \to \rB(\bZ/2)$ --- the possible choices are parameterized by $\H^2(\bZ/2;\bZ/2) \cong \bZ/2$ --- which adjusts how orientation reversal behaves on fermions.

There are also various  topological local structures $\cG$ satisfying $\cG(\bR^d) = \Spec(\SuperVect_\bC)$ but in which part or all of $\pi_{\leq 1}\rO(d)$ acts trivially, analogous to the $\bC$-linear unstructured field theories from \cref{eg.complex field theory}.  We now illustrate a few of the possible choices to emphasize that spin and statistics are not intrinsically linked, even in the presence of Hermiticity.  We will then prove \cref{mainthm} showing that spin and statistics are linked when an extra reflection-positivity hypothesis is imposed.  In order to construct examples of field theories with various topological local structures, we focus on the case when $d=2$, since then we can use the classification of 2-dimensional field theories from \cite{Schommer-Pries:thesis}.

\begin{example}\label{eg.spin but not super}
  A \define{Hermitian spin field theory} is an $\bR$-linear field theory with local structure $\Spins \times_{\bZ/2} \Spec(\bC)$. 
        Unpacking the definition, a Hermitian spin field theory is a non-super $\bC$-linear spin field theory such that orientation reversal agrees with complex conjugation.  In terms of simultaneously-spin-and-super field theories,  $\twist$ acts nontrivially but $(-1)^f$ acts trivially.
  
  Two-dimensional $\bC$-linear spin field theories in $\Alg$ are classified by finite-dimensional complex semisimple algebras $A$ equipped with a trivialization $\varphi: A^* \otimes_A A^* \isom A$ of $A$-$A$ bimodules, where $A^*$ denotes the linear dual bimodule to $A$, such that the two maps $\varphi \otimes \id: A^* \otimes_A A^* \otimes_A A^* \isom A \otimes_A A^* = A^*$ and $\id \otimes \varphi : A^* \otimes_A A^* \otimes_A A^* \isom A^* \otimes_A A = A^*$ agree.  The Hermiticity requirement unpacks to having a
  ($\bC$-antilinear) \define{stellar} structure, i.e.\ a 
   Morita equivalence $A^\op \cong \bar A$, where $\bar A$ is the complex-conjugate algebra, satisfying satisfying certain requirements \cite[Section 3.8.6]{Schommer-Pries:thesis}.  
   Stellar structures are the Morita-equivariant version of $*$-structures, and any $*$-structure defines a stellar structure.
  Hermiticity requires that  $\varphi$  be real.
  
  For example, we can take $A = \bC$ with its standard $\ast$-algebra structure, and choose the trivialization $\varphi: \bC = \bC^* \otimes_\bC \bC^* \isom \bC$ to be multiplication by $-1$.  
  Either trivialization $\pm \sqrt{-1} : \bC^* \to \bC$ presents the $\bC$-linear field theory defined by $A$ as the 
  underlying spin field theory of an oriented field theory over $\bC$.
  But as a Hermitian spin theory, the field theory defined by $A$ is fundamentally spin, since neither $\pm\sqrt{-1}$ is real.
\end{example}

\begin{example}\label{eg.super but not spin}
  A \define{Hermitian super field theory} is an $\bR$-linear field theory with local structure $\Or \times_{\bZ/2} \Spec(\SuperVect_\bC) \cong \Her \times \Spec(\SuperVect_\bR)$,  
  i.e.\ an
  oriented field theory valued in $\SuperVect_\bC$ such that orientation reversal agrees with complex conjugation. 
   In terms of simultaneously-spin-and-super field theories,  $(-1)^f$ acts nontrivially but $\twist$ acts trivially.
  
  Two-dimensional Hermitian super field theories are classified by  symmetric Frobenius stellar superalgebras.  In particular, every symmetric Frobenius $\ast$-superalgebra determines a Hermitian super field theory.  Consider the complex superalgebra $\Cliff(2) = \bC\langle x,y\rangle / (x^2 = y^2 = 1, \, [x,y]=0)$, where $x$ and $y$ are odd.  It admits a $\ast$-structure in which $x^* = x\sqrt{-1}$ and $y^* = y\sqrt{-1}$.  Then $xy$ is imaginary and even, and $\Cliff(2)$ admits a symmetric Frobenius $\ast$-superalgebra structure in which $\tr(xy) = \sqrt{-1}$ and $\tr(1) = \tr(x) = \tr(y) = 0$.  
  
  As a complex Frobenius superalgebra, $\Cliff(2)$ is Morita-equivalent to $\bC$, and so the $\bC$-linear oriented super field theory defined by $\Cliff(2)$ is the super-ification of a purely bosonic theory.  But the Morita equivalence $\Cliff(2) \simeq \bC$ is not compatible with the stellar structure, and so the corresponding Hermitian super field theory defined by $\Cliff(2)$ is fundamentally super.
\end{example}

\begin{example}\label{twisted frobenius structure}
  Two-dimensional spin-statistics field theories are classified by finite-dimensional semisimple ``twisted-symmetric'' Frobenius superalgebras.  Specifically, let $A$ be a finite-dimensional semisimple superalgerba arising as $Z(\pt)$ for some two-dimensional field theory.  Then $360^\circ$ rotation acts by the dual bimodule $Z(\twist) = {_A^{}A^*_A}$.  Let $_A^{}(-1)^f_A$ denote the bimodule $A$ with actions $a \triangleright m \triangleleft b = am(-1)^{|b|}b$; it is the bimodule corresponding to the algebra automorphism $(-1)^f : A \to A$.  The spin-statistics data ``$\twist = (-1)^f$'' then corresponds to a bimodule isomorphism $\phi : {_A^{}A^*_A} \isom {_A^{}(-1)^f_A}$.
  
  Consider the trace $\tr (a) = \langle \phi^{-1}(1_A),a\rangle$, where $\langle,\rangle : A^* \otimes A \to \bC$ denotes the canonical pairing.  This trace is not symmetric.  In a symmetric  Frobenius superalgebra, the trace should satisfy $\tr(ab) = (-1)^{|a|\cdot |b|}\tr(ba)$.  Instead, the trace pairing above satisfies $\tr(ab) = (-1)^{|a|\cdot (|b|+1)}\tr(ba) = \tr(ba)$, where the second equality follows from the fact that $\tr$, being an even map, vanishes on odd elements.  Thus not the superalgebra $A$ but rather the underlying non-super algebra $\Forget(A)$ is symmetric Frobenius.
  
  Real spin-statistics field theories are classified by twisted-symmetric Frobenius superalgebras in $\SuperVect_\bR$.  Hermitian spin-statistics field theories are classified by twisted-symmetric Frobenius stellar superalgebras in $\SuperVect_\bC$, where the isomorphism $\phi$ is real.
  
  The Clifford algebras $\Cliff(n) = \bC\langle x_1,\dots,x_n\rangle / ([x_j,x_k] = 2\delta_{jk})$ admit twisted-symmetric Frobenius $\ast$-superalgebra structures.   As in \cref{eg.super but not spin}, we can give $\Cliff(n)$ a $\ast$-structure by declaring $x_j^* = x_j\sqrt{-1}$.    
  When $n$ is odd, there is an isomorphism of superalgebras $\Cliff(n) \cong \Cliff(1) \otimes \mathrm{Mat}_\bC(2^{(n-1)/2})$, where $\mathrm{Mat}_\bC(m)$ is the purely-even algebra of $m\times m$ complex matrices, and so we can define the trace $\tr: \Cliff(n) \to \bC$ to be the matrix trace on the even part (and to vanish on the odd part).   
  This $\tr$ is twisted-symmetric and real and so defines a two-dimensional spin-statistics Hermitian field theory.  When $n$ is even, the isomorphism $\Forget(\Cliff(n)) \cong \mathrm{Mat}_\bC(2^{n/2})$ defines a twisted-symmetric Frobenius structure on $\Cliff(n)$.  When $n$ is even, $\Cliff(n)$ also admits a non-twisted symmetric Frobenius structure; \cref{eg.super but not spin} describes the case $n=2$.
\end{example}

\begin{example}\label{twisted versions of above}
  A \define{twisted-Hermitian spin field theory} is like a Hermitian spin field theory except that rather than the canonical action of $\bZ/2$ on $\Spins$, we  twist the action by the nontrivial map $\bZ/2 \to \rB(\bZ/2)$.  This unpacks to the requirement that the trivialization $\varphi$ in \cref{eg.spin but not super} be pure-imaginary.

  A \define{twisted-Hermitian super field theory} is like a Hermitian super field theory except that rather than the canonical action of $\bZ/2$ on $\SuperVect_\bC$, we  twist the action by the nontrivial map $\bZ/2 \to \rB(\bZ/2)$.  These are classified not by symmetric Frobenius stellar superalgebras, but by symmetric Frobenius \define{twisted-stellar} superalgebras.  These are defined analogously to stellar superalgebras but with one modification.  For any superalgebra $A$, consider the superalgebra $A'$ defined by $x \cdot' y = (-1)^{|x|\cdot |y|}xy$.  A stellar structure on $A$ includes a Morita equivalence between the opposite superalgebra $A^\op$ and the complex conjugate superalgebra $\bar A$.  A twisted-stellar structure instead makes $A^\op$ equivalent to $\bar A'$.  A special case is that of twisted-$\ast$-superalgebras.  In a $\ast$-superalgebra, $x \mapsto x^*$ must be an algebra anti-automorphism, which in $\SuperVect_\bC$ means that $(xy)^* = (-1)^{|x|\cdot|y|}y^*x^*$.  In a twisted-$\ast$-superalgebra, we have instead $(xy)^* = y^*x^*$ for elements of arbitrary parity.
  Examples of twisted-$\ast$ superalgerbas include $\Cliff(n)$ for arbitrary $n$ with $x_i^* = x_i$.  
  
  The nontrivial automorphism of $\bZ/2 \times \rB(\bZ/2)$ mentioned in \cref{canonical splitting} defines a second $\bZ/2 \times \rB(\bZ/2)$-torsor over $\Spec(\bR)$ with total space $\Spec(\SuperVect_\bC)$.  The corresponding topological local structure controls \define{twisted-Hermitian spin-statistics field theories}.  The twisted-$\ast$-superalgerbas $\Cliff(n)$ with their twisted-symmetric Frobenius structures from   \cref{twisted frobenius structure} provide examples of 
  twisted Hermitian spin-statistics field theories.

  \define{Twisted-real} spin-statistics field theories are classified by twisted-symmetric Frobenius algebra objects in the category $\SuperVect_\bH$ from \cref{supervectH}.
 \end{example}

We now extend the notion of ``reflection-positivity'' from 
\cref{defn.reflection positive,defn.reflection positive for oriented}
to the \'etale-locally-spin case.  Following the physics literature, and in disagreement with \cite{FreedHopkins}, we declare that reflection-positivity of an extended field theory can be detected in codimension-one:

\begin{definition}\label{defn.extended reflection positive}
  An extended unstructured field theory $Z : \Bord_{d-2,d-1,d} \to \Alg_\bR$ is \define{reflection-positive} if its restriction $Z|_{\Bord_{d-1,d}} : \Bord_{d-1,d} \to \Vect_\bR$ to an unextended field theory is reflection-positive in the sense of \cref{defn.reflection positive}, i.e.\ if for every closed $(d-1)$-dimensional manifold~$N$, the symmetric pairing $Z(N \times \pairing) : Z(N)^{\otimes 2} \to \bR$ is positive-definite.
\end{definition}

In \cref{defn.reflection positive for oriented} we defined reflection-positivity for \'etale-locally-oriented field theories in terms of integration over the space of ``\'etale-local orientations.''  We now extend that logic to \'etale-locally-spin field theories.
Consider first the case when $Z$ is spin.  For any manifold $M$, $\Spins(M)$ is a finite groupoid, and so \cite{MR1852152} defines an integration map $\int_{\Spins(M)} : \cO(\Spins(M)) \to \bR$ by $\int_{\Spins(M)}f = \sum_{x\in \pi_0\Spins(M)} f(x) / |\pi_1(\Spins(M), x)|$.   When $V$ is a  bundle over $\Spins(M)$, $\int_{\Spins(M)}V$ is the space of coinvariants of $V$.  

\begin{example}\label{eg.int over spins}
  Given a two-dimensional non-Hermitian spin field theory $Z$ corresponding to the algebra $Z(\pt) = A$ and trivialization $\varphi: A^*\otimes_A A^* \isom A$, one can compute the unstructured field theory $\int_{\Spins}Z$ in two steps.  First, one can integrate over the fibers of the projection $\Spins \to \Or$.  The corresponding oriented field theory $\int_{\Spins/\Or}Z$ is controlled by the symmetric Frobenius algebra $B = A \oplus A^*$ with multiplication is $(a \oplus \alpha)\cdot (b \oplus \beta) = (ab + \varphi(\alpha \otimes \beta)) \oplus (a\beta + \alpha b)$, and the Frobenius structure is $\tr(a\oplus\alpha) = \alpha(1)$.  Second, one can integrate over the choice of orientation,  producing the unstructured field theory controlled by $B \oplus B^{\op}$ with the obvious algebraic $\ast$-structure.
\end{example}

The construction ``$\int_{\Spins}$'' makes sense for any \'etale-locally-spin field theory: if $Z$ has local structure $\cG$ where $\cG(\pt)$ is a categorified $\bZ/2 \times \rB(\bZ/2)$-torsor, then the base change $Z_\cG$ of $Z$ along $\cG(\pt) \to \Spec(\bR)$ is a Galois-equivariant spin field theory over $\cG(\pt)$; thus $\int_{\Spins} Z_\cG$ is a Galois-equivariant unstructured field theory and so descends to  $\Spec(\bR)$.

\begin{example}\label{eg.int over spinstats}
  Suppose that $Z$ is a two-dimensional spin-statistics field theory, either Hermitian or oriented.  In order to treat both oriented and Hermitian field theories, we first study the $\bC$-linear spin-statistics field theory $Z_\bC = Z \otimes_\bR \bC$.  
  
  As in \cref{twisted frobenius structure}, $Z_\bC$ is determined by a finite-dimensional semisimple $\bC$-linear  superalgebra $A$ together with a bimodule isomorphism 
  $\phi : {_A^{}A^*_A} \isom {_A^{}(-1)^f_A}$.  Let $\tilde Z$ denote the $\SuperVect_\bC$-valued spin field theory determined by $A$ together with the isomorphism $\varphi = \phi \otimes \phi : A^* \otimes_A A^* \isom (-1)^f \otimes_A (-1)^f \cong A$.  
 Integrate $\tilde Z$ to a $\SuperVect_\bC$-valued oriented field theory $\int_{\Spins/\Or} \tilde Z$ controlled by the superalgebra algebra $B = A \oplus A^*$.  Then $\tilde Z$ canonically descends to a non-super $\bC$-linear oriented field theory $\int_{\Spins/\Or} Z_\bC$.  Indeed, the isomorphism $\phi : A^* \isom (-1)^f_A$ identifies $B$ with the semidirect product for the parity reversal action $A \rtimes \bZ/2 = A \oplus A\epsilon$ where $\epsilon = \epsilon^{-1}$ is even and $\epsilon a = (-1)^{|a|}a\epsilon$.  In particular, the bimodule $^{}_B(-1)^f_B$ is canonically trivialized by $b \mapsto b\epsilon$.  Let $\Forget(B)$ denote the underlying non-super algebra of $B$.
  The trivialization ${_B^{}(-1)^f_B} \cong {_BB_B}$ determines a Morita equivalence, namely $B / (1-\epsilon) \oplus \Pi B / (1+\epsilon)$, between $B$ and $\Forget(B)$.  The non-super functor $\int_{\Spins/\Or}Z_\bC$ assigns $\Forget(B)$ to the point.
  
  Finally, because the entire construction is equivariant under complex conjugation, if $Z$ was real, then  $\int_{\Spins/\Or}Z_\bC$ naturally descends to a real oriented field theory, and if $Z$ was Hermitian, then  $\int_{\Spins/\Or}Z_\bC$ is naturally Hermitian.  Let us describe the Hermitian case, as it is the more interesting one.  In terms of algebras, if $Z$ was Hermitian, then $A$ is stellar.
  By declaring that $\epsilon$ is real, $B$ also becomes stellar, and hence so too is the Morita-equivalent purely even algebra $\Forget(B)$.  This stellar structure defines $\int_{\Spins/\Or}Z$ as a Hermitian field theory.
  In most examples, the stellar structure on $A$  comes from a $*$-structure.  In this case, $B$ is also $*$. 
   After tracing through the equivalences, one finds that the induced $*$-structure on $\Forget(B)$ is $b \mapsto b^*\epsilon^{|b|}$.
\end{example}

\begin{remark}\label{remark.categorical}
  One can also understand \cref{eg.int over spinstats} in terms of categories of modules.  The Morita class of the superalgebra $A = \tilde Z(\pt)$ is determined by the supercategory $\cA = \cat{SuperMod}_A$. $\tilde Z(\twist)$ defines an action of $\bZ/2$ on $\cA$, and $\cB = \cat{SuperMod}_B$ is the supercategory of fixed points for this action.  Being  supercategories, $\cA$ and $\cB$ carry endo-superfunctors $(-1)^f_\cA$ and $(-1)^f_\cB$ which are the identity on objects and even morphisms but act by $(-1)^f$ on odd morphisms.  The spin-statistics data $\tilde Z(\twist) \cong (-1)^f_\cA$ provides a trivialization of $(-1)^f_\cB$.  This is precisely the data needed to factor $\cB \simeq \cB_{\mathrm{ev}} \boxtimes \SuperVect_\bC$, where $\cB_{\mathrm{ev}}$ is the plain category consisting of the ``even'' objects of $\cB$, i.e.\ objects for which the trivialization $(-1)^f_\cB \cong \id$ acts as the identity.
  
  The Morita equivalence between $B$ and $\mathrm{Forget}(B)$ in  \cref{eg.int over spinstats} identifies $\cB_{\mathrm{ev}}$ with $\Mod_{\mathrm{Forget}(B)}$.  A straightforward calculation  shows that $\cB_{\mathrm{ev}}$ is also the underlying non-super category $\cA_0$ of $\cA$, i.e.\ the one with the same objects and even morphisms but with odd morphisms forgotten.  Since the restriction to $\cA_0$ of $(-1)^f_\cA$ is trivial, and since we started with an isomorphism of superfunctors $(-1)^f_\cA \cong \twist$, on the category $\cA_0$ we have $\twist \cong \id$.  This is another way to see that $\cA_0 = \cB_{\mathrm{ev}}$ defines an oriented field theory.
\end{remark}

\begin{definition}\label{defn.reflection positive for spin}
  An \'etale-locally-spin field theory $Z$ is \define{reflection-positive} if the unstructured field theory $\int_{\Spins}Z = \int_{\Or}\int_{\Spins/\Or} Z$ is reflection-positive.
\end{definition}

We can now prove \cref{mainthm}, which asserts that all extended \'etale-locally-spin reflection-positive field theory are  Hermitian and satisfy spin-statistics.

\begin{proof}[Proof of \cref{mainthm}]
  An \'etale-locally-spin field theory either satisfies spin-statistics or is spin-but-not-super.  It suffices to show that if $Z$ is a non-zero spin-but-not-super field theory then it is not reflection-positive; Hermiticity will follow from \cref{reflection positive oriented}.  (A field theory is \define{zero} if it sends to the zero object all non-empty cobordisms.  The zero field theory is vacuously reflection-positive and makes sense for all topological local structures.)
  
  Suppose that $Z$ is a non-zero spin-but-not-super field theory and consider the $\bC$-linear spin-but-not-super field theory $Z_\bC = Z \otimes_\bR \bC$.  Let $P$ be a connected oriented $(d-2)$-dimensional manifold and let $\Spins/\Or(P)$ denote the groupoid of spin structures on $P$ compatible with the chosen orientation.  Since $Z$ is non-zero, we can find $P$ such that for at least one $\sigma \in \Spins/\Or(P)$, $Z(P,\sigma) \neq 0$.  Then in particular $\Spins/\Or(P) \neq \emptyset$ and so $\Spins/\Or(P)$ is a torsor for $\rB(\bZ/2) \times \rH^1(P;\bZ/2)$.
  
  Each choice of $\sigma \in \Spins/\Or(P)$ determines a dimensional reduction of $Z_\bC$ to the two-dimensional $\bC$-linear spin-but-not-super field theory $Z_\bC( - \times (P,\sigma))$.  By the classification of two-dimensional field theories \cite{Schommer-Pries:thesis}, $A = Z_\bC(\pt \times (P,\Sigma))$ is a finite-dimensional semisimple algebra over $\bC$, and so up to Morita equivalence we can assume $A = \bC^{\oplus n}$ for some $n$. A bimodule isomorphism $A^* \otimes_A A^* \isom A$ cannot permute the direct summands, and so the field theory $Z_\bC(-\times (P,\sigma))$ is equivalent to a direct sum $\bigoplus_{i=1}^n Y_{\sigma}^{(i)}$ of complex-linear spin field theories each of which satisfies $A^{(i)} = Y_\sigma^{(i)}(\pt) = \bC$.
  
  The two-dimensional spin-but-not-super field theory $Y_\sigma^{(i)}$ then satisfies $\int_{\Spins/\Or}Y_\sigma^{(i)}(\pt) = A^{(i)} \oplus (A^{(i)})^* =  \bC[x]/(x^2 = 1)$ with $\tr(a+bx) = b$, and the complex Hilbert space is $\int_{\Spins/\Or}Y_\sigma^{(i)}(S^1) = \bC^2$ with purely off-diagonal inner product.  Thus $\int_{\Spins/\Or}Z(P\times S^1) = \int_{\sigma \in \Spins/\Or(P)} \bigoplus_i \bC^2$ is a non-zero direct sum of Hilbert spaces with purely off-diagonal inner product.  Such an inner product cannot be positive-definite.
\end{proof}

\begin{example}
  The Hermitian spin-statistics field theory $Z_n$ defined by $\Cliff(n)$ from \cref{twisted frobenius structure} is reflection-positive. Indeed, when $n$ is odd,  \cref{eg.int over spinstats} implies that the Hermitian field theory $\int_{\Spins/\Or}Z_n$ is controlled by  $\Forget(\Cliff(n) \rtimes \bZ/2) \cong \mathrm{Mat}_\bC(2^{(n+1)/2})$.  As discussed before \cref{defn.reflection positive for oriented}, the Hilbert space $\bigl(\int_{\Spins}Z_n\bigr)(S^1)$ is then the underlying real vector space of $\bigl(\int_{\Spins/\Or}Z_n\bigr)(S^1) = \bC$ equipped with the real part of its Hermitian pairing, which comes in turn from the $*$-structure on $\mathrm{Mat}_\bC(2^{(n+1)/2})$.  
  But $\langle v , v \rangle = |v|^2 \langle 1,1\rangle = |v|^2\tr(1) = |v|^2 2^{(n+1)/2} > 0$, so $Z_n$ is reflection-positive.
  When $n$ is even, $\bigl(\int_{\Spins}Z_n\bigr)(S^1) \cong \bC^{\oplus 2}$ with its positive-definite Hermitian form, where the first copy of $\bC$ comes from ``a boson on $S^1$ with its trivial spin structure'' and the second from ``a fermion on $S^1$ with its nontrivial spin structure.''
  
  When $n$ is even, $\Cliff(n)$ also admits a symmetric Frobenius structure, and so defines a Hermitian non-spin super field theory $Z'_n$.  We can mimic \cref{remark.reflection positive for complex} and integrate over $\Spec(\SuperVect_\bC)$.  The corresponding Hilbert space $\bigl(\int_{\Spec(\SuperVect_\bC)}Z'_n\bigr)(S^1)$ is again a copy of $\bC^{\oplus 2}$, but this time with the indefinite Hermitian inner product.
\end{example}

\section{Extension to higher categories} \label{section.general nonsense}

This last section explains how to extend the ideas in this paper to the higher-categorical setting championed by \cite{Lur09}.
We will assume familiarity with $(\infty,n)$-categories and give only an outline of the necessary constructions.  Following the by-now standard notation in the $\infty$-categorical literature, we let $\Spaces$ denote the $\infty$-category of  topological spaces.
For the remainder of this paper, let  $\Man_d$ denote the $(\infty,1)$-category coming from the topological category of $d$-dimensional smooth manifolds and local diffeomorphisms.  Given an $(\infty,1)$-category $\sX$ with  limits, a \define{topological local structure on $d$-dimensional manifolds valued in $\sX$} is a sheaf $\cG: \Man_d \to \sX$. 
We will not specify precisely the meaning of ``sheaf''; one version is spelled out in \cite{AyalaThesis}.  
(Although the paper \cite{AyalaThesis} begins with ``geometric'' local structures, its Main Theorem asserts that the cobordism category it constructs from a geometric local structure $\cF$ depends only on the corresponding topological local structure $\tau \cF$.)
  We will care most about the case when $\sX$ is an $\infty$-topos, for example the $\infty$-topos of sheaves of spaces on a site like $\Sch_\bR$ or $\CatAffSch_\bR$. 

Generalizing \cref{set theory cob hyp,gpoid cob hyp}, the following standard fact follows from the existence of good open covers together with the homotopy equivalence $\rO(d) \simeq \hom_{\Man_d}(\bR^d,\bR^d)$, c.f.\ \cite{AF2012}:

\begin{lemma} \label{cobord hyp for sheaves}
  The $(\infty,1)$-category of $\sX$-valued topological local structures on $d$-dimensional manifolds is equivalent to the $(\infty,1)$-category $\sX^{\rO(d)}$ of $\sX$-objects equipped with an action by the topological group $\rO(d)$, the equivalence being given by sending a sheaf $\cG : \Man_d \to \sX$ to $\cG(\bR^d) \in \sX$. \qedhere
\end{lemma}

 The sheaf corresponding to an object $X \in \sX^{\rO(d)}$ can be constructed as follows.  Given any $d$-manifold $M \in \Man_d$,
  define $X(M) = \maps_{\rO(d)}(\mathrm{Fr}_M,X) \in \sX$, where
    $\mathrm{Fr}_M \to M$ denotes the frame bundle,
    $\maps_{\rO(d)}$ denotes $\rO(d)$-equivariant maps, and $\sX$ is tensored over $\Spaces$ since it is an $\infty$-topos.  Then $X(\bR^d) \simeq X$ in $\sX^{\rO(d)}$.
  A special case 
   is when $X \in \sX$ is equipped with the trivial $\rO(d)$ action.  Then $X(M) \simeq \maps(M,X)$ is the \define{classical topological sigma-model with target $X$}.

\begin{example}\label{tangential structure}
  Let $G \to \rO(d)$ be a map of topological groups.  A \define{$G$-tangential structure} on a $d$-dimensional manifold $M$ is a $G$-principal bundle $P \to M$ with an equivalence $P \times_G \rO(d) \simeq \mathrm{Fr}_M$ of $\rO(d)$-bundles.  The sheaf $\Man_d \to \Spaces$ of $G$-tangential structures is classified by the quotient $\rO(d) / G$ with its natural $\rO(d)$-action.
\end{example}

We now explain how to build an $(\infty,d)$-category $\Bord_d^\cG = \Bord_{0,\dots,d}^\cG$ of ``$\cG$-structured bordisms'' for each topological local structure $\cG$.  For a suitable target $\sV$, a \define{fully-extended $\cG$-structured quantum field theory} will be a symmetric monoidal functor $\Bord_d^\cG \to \sV$; the precise statement is in \cref{defn.G-structured theories}.  We let $\Bord_d$ denote the ``unstructured'' bordism category whose construction is thoroughly outlined in \cite{Lur09}, and for which all details have been provided by \cite{ClaudiaDamien1}.  We will not review the construction of $\Bord_d$ itself.

Let $\sY$ be an $(\infty,1)$-category with finite limits; for example, $\sY = \sX$ an $\infty$-topos.  For each $d$, the paper \cite{Haugseng-Spans} builds from $\sY$ a symmetric monoidal $(\infty,d)$-category $\Spans_d(\sY)$.  (The case when $\sY = \Spaces$ is outlined, under the name $\cat{Fam}_d$, in \cite{Lur09}.)  
The objects of $\Spans_d(\sY)$ are are those of $\sY$, but a 1-morphism from $X$ to $Y$ in $\Spans_d(\sY)$ is a \define{span} $X \leftarrow A \rightarrow Y$ in $\sY$, and higher morphisms are spans-between-spans.
Following \cite{Haugseng-Spans}, we call symmetric monoidal functors $\Bord_d \to \Spans_d(\sY)$ \define{classical} (unstructured, fully extended) field theories valued in $\sY$.

Every $\sY$-valued topological local structure $\cG$ determines a classical field theory $\widetilde\cG$ (and the celebrated Cobordism Hypothesis of \cite{Lur09} implies that all classical field theories arise from topological local structures).  Indeed, given a $k$-dimensional manifold $M$ for $k\leq d$, set $\widetilde\cG(M) = \cG(M \times \bR^{d-k})$; if $M$ has boundary, first glue on a ``collar'' $M \leadsto M \cup_{\partial M} (\partial M \times \bR_{\geq 0})$.  Then if $M$ is a cobordism from $N_1$ to $N_2$, the restriction maps $\cG(M) \to \cG(N_1)$ and $\cG(M) \to \cG(N_2)$ make $\cG(M)$ into a span, and
functoriality for the assignment $\widetilde\cG : \Bord_d \to \Spans_d(\sY)$ follows from the sheaf axiom for~$\cG$.

\begin{remark}
  In the model of $\Bord_d$ from \cite{ClaudiaDamien1}, $k$-morphisms are not  $k$-dimensional manifolds, but rather  $d$-dimensional manifolds properly submersed over $\bR^{d-k}$.  When using that model, one can directly define  $\widetilde\cG : \Bord_d \to \Spans_d(\sY)$ simply as $\widetilde\cG(M) = \cG(M)$.
\end{remark}

Let $\pt\in \sY$ denote the terminal object and $\sY_{\pt/}$ the ``undercategory'' of \define{pointed objects} $\pt \to X$ in $\sY$. 
The logic of \cite{Lur09} is to construct $\Bord_d^\cG$ for $\cG$ a $\Spaces$-valued topological local structure and then observe that there is a pullback square of symmetric monoidal $(\infty,d)$-categories, where the vertical arrows are the obvious forgetful functors:
$$ \begin{tikzpicture}
   \path node (UL) {$\Bord_d^\cG$} +(3.5,0) node (UR) {$\Spans_d(\Spaces_{\pt/})$} +(0,-1.5) node (LL) {$\Bord_d$} +(3.5,-1.5) node (LR) {$\Spans_d(\Spaces)$}
     +(.5,-.5) node {$\ulcorner$};
   \draw[->] (UL) -- (UR); \draw[->] (UL) -- node[auto,swap] {\scriptsize {Forget} $\cG$} (LL); \draw[->] (LL) -- node[auto] {$\scriptstyle \cG$} (LR); \draw[->] (UR) -- node[auto] {\scriptsize {Forget} the pointing} (LR);
  \end{tikzpicture}$$
Indeed, a $\cG$-structured manifold $M$ is nothing but a manifold $M$ together with a pointing of the space $\cG(M)$.
We will reverse the logic and interpret the above pullback square as the definition of $\Bord_d^\cG$.  Some care must be taken when replacing $\Spaces$ by an $\infty$-topos $\sX$, as in general very few objects $X \in \sX$ admit ``global'' points $\pt \to X$.  The correct approach is to work with symmetric monoidal $(\infty,d)$-categories ``internal to $\sX$''; for the definition, see \cite{
Haugseng-Spans,Li-Bland2015}.

\begin{definition}\label{defn.G-structured theories}
  Let $\sX$ be an $\infty$-topos.  By \cite[Theorem 4.3]{Li-Bland2015}, the symmetric monoidal $(\infty,d)$-category $\Spans_d(\sX)$ constructed in \cite{Haugseng-Spans} underlies an internal symmetric monoidal $(\infty,d)$-category in $\sX$, which in an abuse of notation we will also call $\Spans_d(\sX)$; the same argument implies also that $\Spans_d(\sX_{\pt/})$ is naturally an internal symmetric monoidal $(\infty,d)$-category in $\sX$.  Via the unique topos map $\Spaces \to \sX$, also view $\Bord_d$ as an internal symmetric monoidal $(\infty,d)$-category in $\sX$.
  
  Let $\cG : \Man_d \to \sX$ be an $\sX$-valued topological local structure and $\widetilde\cG : \Bord_d \to \Spans_d(\sX)$ the corresponding classical field theory.  It extends canonically to a functor of internal symmetric monoidal $(\infty,d)$-categories.  The $(\infty,d)$-category $\Bord_d^\cG$ of \define{$\cG$-structured bordisms} is by definition the following pullback of {internal} symmetric monoidal $(\infty,d)$-categories:
  $$ \begin{tikzpicture}
   \path node (UL) {$\Bord_d^\cG$} +(3.5,0) node (UR) {$\Spans_d(\sX_{\pt/})$} +(0,-1.5) node (LL) {$\Bord_d$} +(3.5,-1.5) node (LR) {$\Spans_d(\sX)$}
     +(.5,-.5) node {$\ulcorner$};
   \draw[->] (UL) -- (UR); \draw[->] (UL) -- (LL); \draw[->] (LL) -- node[auto] {$\scriptstyle \cG$} (LR); \draw[->] (UR) -- node[auto] {\scriptsize {Forget} the pointing} (LR);
  \end{tikzpicture}$$
  
  Let $\sV$ be a symmetric monoidal $(\infty,d)$-category internal to $\sX$.  A \define{$\cG$-structured field theory valued in $\sV$} is a functor $\Bord_d^\cG \to \sV$ of internal symmetric monoidal $(\infty,d)$-categories.
\end{definition}

The need to work with internal categories in \cref{defn.G-structured theories} is in some sense unavoidable --- ``functors internal to $\sX$'' is the appropriate language with which to impose that a functor be ``smooth'' for families parameterized by objects of $\sX$.  But one can also describe $\cG$-structured field theories ``externally'' in terms of the lifting problems in \cref{schemes structured field theory,bicategorical description of G-structures}.
Given an $\infty$-topos $\sX$ and a symmetric monoidal $(\infty,d)$-category $\sV$ internal to $\sX$, the papers \cite{Haugseng-Spans,Li-Bland2015} construct a symmetric monoidal $(\infty,d)$-category $\Spans_d(\sX;\sV)$ whose $k$-morphisms are ``bundles of $k$-morphisms in $\sV$ over $k$-fold spans in $\sX$.''     Such a notion makes sense exactly because $\sV$ is  internal to $\sX$: by definition, a \define{bundle of $k$-morphisms in $\sV$ over $X\in \sX$} is a map from $X$ to the $\sX$-object of $k$-morphisms in $\sV$.  After unpacking adjunctions, one finds:

\begin{proposition}\label{bundle result}
  Let $\sX$ be an $\infty$-topos, $\cG : \Man_d \to \sX$ a topological local structure, and $\sV$ a symmetric monoidal $(\infty,d)$-category internal to $\sX$.  Then the data of a $\cG$-structured field theory $\Bord_d^\cG \to \sV$ is the same as the data of a lift:
  
  \mbox{}\hfill $ \begin{tikzpicture}[baseline=(LR.base)]
   \path coordinate (UL) +(3,0) node (UR) {$\Spans_d(\sX;\sV)$} +(0,-1.5) node (LL) {$\Bord_d$} +(3,-1.5) node (LR) {$\Spans_d(\sX)$}
    ;
   \draw[->] (LL) -- node[auto] {$\scriptstyle \widetilde\cG$} (LR); \draw[->] (UR) -- node[auto] {\scriptsize {Forget} the $\sV$-data} (LR);
   \draw[->,dashed] (LL) -- (UR);
  \end{tikzpicture} $ \qedhere
\end{proposition}

\begin{corollary} \label{G-structured cobordism hypothesis}
  Let $\sX$ be an $\infty$-topos and $\sV$ an internal-to-$\sX$ symmetric monoidal $(\infty,d)$-category with duals in the sense of \cite{Haugseng-Spans}.  Let $\cG : \Man_d \to \sX$ be a topological local structure, and $\cG(\pt) = \cG(\bR^d)$ the corresponding object in $\sX^{\rO(d)}$.  Assuming the Cobordism Hypothesis, $\cG$-structured field theories valued in $\sV$ are classified by $\rO(d)$-equivariant bundles of $\sV$-objects over $\cG(\pt)$. \qedhere
\end{corollary}

We conclude by extending the examples from this paper.  Note that under \cref{cobord hyp for sheaves}, the sheaves $\Or$ and $\Spins$ of orientations and spin structures correspond, respectively, to the actions of $\rO(d)$ on the $0$- and $1$-truncations $\pi_{\leq 0}\rO(\infty)$ and $\pi_{\leq 1}\rO(\infty)$, or equivalently to the trivial torsors for these groups.  Any $\infty$-topos $\sX$ admits a notion of ``torsor'' for topological groups: $X\in \sX^G$ is a \define{$G$-torsor} if the map $G \times X \to X \times X$, $(g,x) \mapsto (gx,x)$ is an equivalence.  An $\sX$-valued topological local structure $\cG : \Man_d\to \sX$ is \define{locally $\Or$} (resp.\ \define{locally $\Spins$}) if $\cG(\bR^d)$ is a torsor for $\pi_{\leq 0}\rO(\infty)$ (resp.\ $\pi_{\leq 1}\rO(\infty)$).  Suppose $\sX$ is the $\infty$-topos of sheaves (valued in $\Spaces$) on some site (with some subcanonical topology) containing the category $\cat{AffSch}_\bR$ of affine schemes over $\bR$.  Then there is a canonical $\sX$-valued topological local structure $\Her: \Man_d \to \sX$ whose value on $\bR^d$ is $\Spec(\bC)$.  If $\sX$ is the $\infty$-topos of sheaves on some site containing $\CatAffSch_\bR$, then similarly there is a canonical topological local structure $\SpinStats : \bR^d \mapsto \Spec(\SuperVect_\bC)$.

For  $\sV$ a suitable target symmetric monoidal $(\infty,d)$-category internal to $\sX$, we can then define \define{Hermitian} and \define{Hermitian spin-statistics} field theories as being $\Her$- and $\SpinStats$-structured field theories in the sense of \cref{defn.G-structured theories}.  Note that the details of the $\infty$-topos $\sX$ are largely irrelevant: given \cref{bundle result}, what matters for Hermitian and spin-statistics field theories are the symmetric monoidal $(\infty,d)$-categories of $X$-points of $\sV$ for $X$ ranging over the possible values $\Spec(\bR), \Spec(\bC), \Spec(\SuperVect_\bC), \dots$ of $\Her$ and $\SpinStats$. 

 One standard criterion for deciding whether a proposed target $\sV$ is suitable is that ``near the top'' $\sV$ should look like $\Vect$.  More precisely, any symmetric monoidal $(\infty,d)$-category $\sV$ determines a symmetric monoidal $(\infty,1)$-category $\Omega^{d-1}\sV$ of endomorphisms of the identity $(d-2)$-morphism on the identity $(d-3)$-morphism on \dots on the unit object in $\sV$.  The passage $\sV \mapsto \Omega^{d-1}\sV$ makes sense also for internal categories.  The ``looks like $\Vect$ near the top'' criterion then says that for $R$ a commutative $\bR$-algebra, the $\Spec(R)$-points of $\Omega^{d-1}\sV$ should be $\Mod_R$, and that for $\cC$ a categorified commutative $\bR$-algebra, the $\Spec(\cC)$-points of $\Omega^{d-1}\sV$ should be $\cC$ itself.
 This assures, for example, that if $Z : \Bord_d^{\Her} \to \sV$ is a fully-extended Hermitian field theory, then its restriction $Z|_{\Bord_{d-1,d}^{\Her}} : \Bord_{d-1,d}^{\Her} = \Omega^{d-1}\Bord_d \to \Omega^{d-1}\sV = \Vect$ unpacks to a Hermitian unextended field theory in the sense of \cref{schemes structured field theory}.

An extended field theory $Z : \Bord_d \to \sV$ is \define{reflection-positive} if the unextended field theory $Z|_{\Bord_{d-1,d}}$ is reflection-positive in the sense of \cref{defn.reflection positive}.  
(This is different from the notion in \cite{FreedHopkins} of ``reflection-positivity'' for extended field theories, which requires extra ``positivity'' data to be specified in high codimension.)
  \Cref{defn.reflection positive for oriented,defn.reflection positive for spin} then apply to extended Hermitian and spin-statistics field theories.
  
One could worry that restricting a field theory just to its top part is too much loss of information.  The following observation is due to Chris Schommer-Pries:

\begin{lemma}\label{lemma.nonzero}
  Let $\sV$ be be some symmetric monoidal $(\infty,d)$-category with a zero object, and $Z : \Bord_d^\cG \to \sV$ be a $\cG$-structured extended field theory for some topological local structure $\cG$.  Suppose that the unextended field theory $Z|_{\Bord_{d-1,d}^\cG}$ is \define{zero} in the sense that it vanishes on all non-empty inputs.
(Symmetric monoidality forces $Z(\emptyset)$ to be the unit object of $\Omega^{d-1}\sV$.)  Then $Z$ is zero.
  \end{lemma}
\begin{proof}
  A $k$-morphism $F$ is zero if and only if its identity $(k+1)$-morphism $\id_F$ is zero.  It therefore suffices to show that for $N$ an arbitrary $\cG$-structured $(d-1)$-dimensional cobordism, $Z(N \times [0,1]) : Z(N) \to Z(N)$ is the zero $d$-morphism.  But the $\cG$-structured cobordism $N\times [0,1]$ can be factored through $N \sqcup S^{d-1}$ where the sphere $S^{d-1}$ is given the $\cG$-structure that extends to the disk $D^d$:
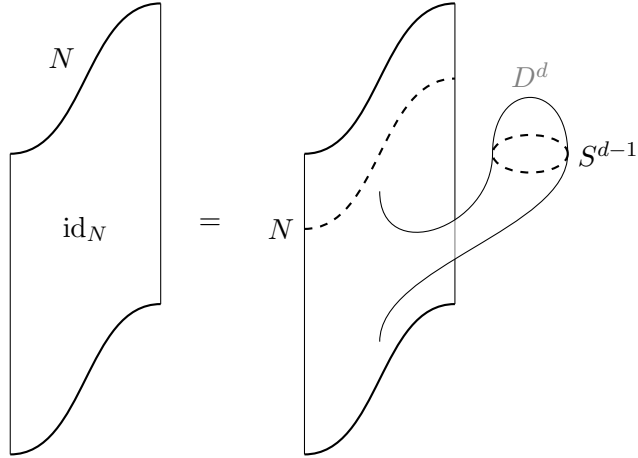
  $$ 
  \begin{tikzpicture}[baseline=(middle)]
    \draw[thick] (0,0) .. controls +(1,0) and +(-1,0) .. (2,2);
    \draw[thick] (0,4) .. controls +(1,0) and +(-1,0) .. node[auto] {$N$} (2,6);
    \path (1,3) node (middle) {$\id_N$};
    \draw (0,0) -- (0,4); \draw (2,2) -- (2,6);
  \end{tikzpicture} 
  \quad = \quad
  \begin{tikzpicture}[baseline=(middle)]
    \draw[thick] (0,0) .. controls +(1,0) and +(-1,0) .. (2,2);
    \draw[thick] (0,4) .. controls +(1,0) and +(-1,0) .. (2,6);
    \draw[thick,dashed] (0,3) .. controls +(1,0) and +(-1,0) ..  (2,5);
    \path (0,3) node[anchor=east] {$N$};
    \path (1,3) coordinate (middle);
    \draw (0,0) -- (0,4); 
    \path (3.5,4) node [anchor=west] {$S^{d-1}$};
    \draw (2,2) -- (2,6);
    \draw[fill opacity = .5,fill=white,text opacity=1] (1,1.5) .. controls +(0,1) and +(0,-1) .. (3.5,4) .. controls +(0,1) and +(0,1) ..  node[auto,swap] {$D^d$} (2.5,4) .. controls +(0,-1) and +(0,-1) .. (1,3.5) (1,1.5);
    \draw[thick,dashed] (3,4) ellipse (.5 and .25);
  \end{tikzpicture} 
  $$
  By assumption, $Z(S^{d-1}) = 0$, since $S^{d-1}$ is closed, and so $Z(N \sqcup S^{d-1}) \cong Z(N) \otimes Z(S^{d-1}) = 0$.  Only a zero morphism can factor through a zero object, and so $Z(N \times [0,1]) = 0$.
\end{proof}

Only the zero field theory is compatible with multiple topological local structures.  \Cref{lemma.nonzero} assures that if a $\cG$-structured fully extended field theory $Z$ is not zero, then neither is its restriction $Z|_{\Bord_{d-2,d-1,d}^\cG}$ to a once-extended theory, and so $Z|_{\Bord_{d-2,d-1,d}^\cG}$ detects the local structure $\cG$.  Along with \cref{mainthm}, we conclude:

\begin{corollary}\label{main thm extended}
  Reflection-positive \'etale-locally-spin fully-extended field theories are necessarily unitary and satisfy spin-statistics. \qedhere
\end{corollary}

\section*{Acknowledgments}

I would like to thank K.\ Costello and D.\ Gaiotto for asking the questions that led to this note, D.\ Freed for emphasizing the importance of positivity in the spin-statistics theorem,  A.\ Chirvasitu and E.\ Elmanto for many discussions about Galois theory and locally presentable categories, and C.\ Schommer-Pries for providing the proof of \cref{lemma.nonzero}.  
I would like particularly to thank the anonymous referee for their considerable engagement with the organization and presentation of the ideas in this paper; their comments and suggestions led to many improvements.
Some of this work occurred while I was a visitor to the Perimeter Institute for Theoretical Physics in Waterloo, Ontario, where I was treated to generous hospitality and a great work environment.  
 Research at Perimeter Institute is supported by the Government of Canada through Industry Canada and by the Province of Ontario through the Ministry of Economic Development \& Innovation.
 This research is also supported by the NSF grant DMS-1304054.


\begin{thebibliography}{BCJF14}

\bibitem[AF12]{AF2012}
David Ayala and John Francis.
\newblock Factorization homology of topological manifolds.
\newblock 2012.
\newblock \arXiv{1206.5522}.

\bibitem[Ati88]{MR1001453}
Michael Atiyah.
\newblock Topological quantum field theories.
\newblock {\em Inst. Hautes {\'E}tudes Sci. Publ. Math.}, (68):175--186 (1989),
  1988.
\newblock \MRnumber{1001453}.

\bibitem[Aya08]{AyalaThesis}
David Ayala.
\newblock {\em Geometric Cobordism Categories}.
\newblock PhD thesis, Stanford University, 2008.
\newblock \arXiv{0811.2280}.

\bibitem[BCJF14]{BCJF2014}
Martin Brandenburg, Alexandru Chirvasitu, and Theo Johnson-Freyd.
\newblock Reflexivity and dualizability in categorified linear algebra.
\newblock 2014.
\newblock \arXiv{1409.5934}.

\bibitem[BD01]{MR1852152}
John~C. Baez and James Dolan.
\newblock From finite sets to {F}eynman diagrams.
\newblock In {\em Mathematics unlimited---2001 and beyond}, pages 29--50.
  Springer, Berlin, 2001.
\newblock \MRnumber{1852152}.

\bibitem[Bir84]{BirdThesis}
Gregory~J. Bird.
\newblock {\em Limits in 2-categories of locally presentable categories}.
\newblock PhD thesis, University of Sydney, 1984.
\newblock Circulated by the Sydney Category Seminar.

\bibitem[CJF13]{MR3097055}
Alexandru Chirvasitu and Theo Johnson-Freyd.
\newblock The fundamental pro-groupoid of an affine 2-scheme.
\newblock {\em Appl. Categ. Structures}, 21(5):469--522, 2013.
\newblock \MRnumber{3097055}. \arXiv{1105.3104}.
  \DOI{10.1007/s10485-011-9275-y}.

\bibitem[CS15]{ClaudiaDamien1}
Damien Calaque and Claudia Scheimbauer.
\newblock A note on the $(\infty,n)$-category of cobordisms.
\newblock 2015.
\newblock \arXiv{1509.08906}.

\bibitem[Del02]{MR1944506}
P.~Deligne.
\newblock Cat{\'e}gories tensorielles.
\newblock {\em Mosc. Math. J.}, 2(2):227--248, 2002.
\newblock Dedicated to Yuri I. Manin on the occasion of his 65th birthday.
  \MRnumber{1944506}.

\bibitem[DSPS14]{DSPS2}
Christopher~L. Douglas, Christopher Schommer-Pries, and Noah Snyder.
\newblock The balanced tensor product of module categories.
\newblock 2014.
\newblock \arXiv{1406.4204}.

\bibitem[Eil60]{MR0125148}
Samuel Eilenberg.
\newblock Abstract description of some basic functors.
\newblock {\em J. Indian Math. Soc. (N.S.)}, 24:231--234 (1961), 1960.
\newblock \MRnumber{0125148}.

\bibitem[FH16]{FreedHopkins}
Daniel~S. Freed and Michael~J. Hopkins.
\newblock Reflection positivity and invertible topological phases.
\newblock 2016.
\newblock \arXiv{1604.06527}.

\bibitem[GK14]{MR3177367}
Nora Ganter and Mikhail Kapranov.
\newblock Symmetric and exterior powers of categories.
\newblock {\em Transform. Groups}, 19(1):57--103, 2014.
\newblock \DOI{10.1007/s00031-014-9255-z}. \MRnumber{3177367}.
  \arXiv{1110.4753}.

\bibitem[Hau14]{Haugseng-Spans}
Rune Haugseng.
\newblock Iterated spans and ``classical'' topological field theories.
\newblock 2014.
\newblock \arXiv{1409.0837}.

\bibitem[Kap15]{Kapranov2015}
Mikhail Kapranov.
\newblock Supergeometry in mathematics and physics.
\newblock 2015.
\newblock \arXiv{1512.07042}.

\bibitem[Kel82]{MR651714}
Gregory~Maxwell Kelly.
\newblock {\em Basic concepts of enriched category theory}, volume~64 of {\em
  London Mathematical Society Lecture Note Series}.
\newblock Cambridge University Press, Cambridge-New York, 1982.
\newblock \MRnumber{651714}.

\bibitem[Kle05]{MR2165457}
Alexander Kleshchev.
\newblock {\em Linear and projective representations of symmetric groups},
  volume 163 of {\em Cambridge Tracts in Mathematics}.
\newblock Cambridge University Press, Cambridge, 2005.
\newblock \DOI{10.1017/CBO9780511542800}. \MRnumber{2165457}.

\bibitem[LB15]{Li-Bland2015}
David Li-Bland.
\newblock The stack of higher internal categories and stacks of iterated spans.
\newblock 2015.
\newblock \arXiv{1506.08870}.

\bibitem[Lew99]{MR1711565}
L.~Gaunce Lewis, Jr.
\newblock When projective does not imply flat, and other homological anomalies.
\newblock {\em Theory Appl. Categ.}, 5:No. 9, 202--250 (electronic), 1999.
\newblock \MRnumber{1711565}.

\bibitem[Lur09]{Lur09}
Jacob Lurie.
\newblock On the classification of topological field theories.
\newblock In {\em Current developments in mathematics, 2008}, pages 129--280.
  Int. Press, Somerville, MA, 2009.
\newblock \MRnumber{2555928}. \arXiv{0905.0465}.

\bibitem[Ost15]{Ostrik2015}
Victor Ostrik.
\newblock On symmetric fusion categories in positive characteristic.
\newblock 03 2015.
\newblock \arXiv{1503.01492}.

\bibitem[Seg04]{MR2079383}
Graeme Segal.
\newblock The definition of conformal field theory.
\newblock In {\em Topology, geometry and quantum field theory}, volume 308 of
  {\em London Math. Soc. Lecture Note Ser.}, pages 421--577. Cambridge Univ.
  Press, Cambridge, 2004.
\newblock \MRnumber{2079383}.

\bibitem[SP11]{Schommer-Pries:thesis}
Christopher~J. Schommer-Pries.
\newblock The classification of two-dimensional extended topological field
  theories.
\newblock 2011.
\newblock \arXiv{1112.1000}.

\bibitem[ST11]{MR2742432}
Stephan Stolz and Peter Teichner.
\newblock Supersymmetric field theories and generalized cohomology.
\newblock In {\em Mathematical foundations of quantum field theory and
  perturbative string theory}, volume~83 of {\em Proc. Sympos. Pure Math.},
  pages 279--340. Amer. Math. Soc., Providence, RI, 2011.
\newblock \MRnumber{2742432}. \arXiv{1108.0189}.

\bibitem[SW64]{MR0161603}
R.~F. Streater and A.~S. Wightman.
\newblock {\em P{CT}, spin and statistics, and all that}.
\newblock W. A. Benjamin, Inc., New York-Amsterdam, 1964.
\newblock \MRnumber{0161603}.

\bibitem[Wat60]{MR0118757}
Charles~E. Watts.
\newblock Intrinsic characterizations of some additive functors.
\newblock {\em Proc. Amer. Math. Soc.}, 11:5--8, 1960.
\newblock \MRnumber{0118757}.

\end{thebibliography}

\end{document}